%% file: Res_V2/main_v2_black.tex
\documentclass[final]{IEEEtran}
\usepackage[T1]{fontenc}    
\usepackage{textcomp}       

\usepackage{mathtools}      
\usepackage{amsfonts}       
\usepackage{amssymb}        
\usepackage{amsthm}         
\usepackage{siunitx}        
\usepackage{optidef}

\usepackage{graphicx}       
\usepackage{color}          
\usepackage{epstopdf}       
\usepackage{psfrag}         

\usepackage{array}          
\usepackage{tabularx}       
\usepackage{booktabs}       
\usepackage{multirow}       
\usepackage{adjustbox}      
\usepackage{longtable}      

\usepackage{caption}        
\usepackage[caption=false,font=footnotesize,labelfont=rm,textfont=rm]{subfig}
\DeclareSubrefFormat{myformat}{#1(#2)} 
\usepackage{stfloats}

\usepackage{algorithm}    
\usepackage{algpseudocode} 

\usepackage{cite}           
\usepackage{url}            
\usepackage[colorlinks=true, allcolors=blue]{hyperref}

\usepackage{stfloats}       
\usepackage{verbatim}       
\usepackage{authblk}        
\usepackage{cuted}          
\usepackage{xpatch}         

\allowdisplaybreaks 

\definecolor{blue}{rgb}{0,0,0}
\definecolor{roundtwoblue}{rgb}{0,0,1}
\newcommand{\roundtwo}[1]{\textcolor{black}{#1}}
\ifdefined\submissionblack
\renewcommand{\roundtwo}[1]{#1}
\fi
\begin{document}
\include{notation_jiajie} 
\title{Integrated Discovery and State-Aware Servicing for Mobile AUVs With UWOC: Modeling and Performance Analysis}

\author{
Qiyu~Ma,~\IEEEmembership{Student Member,~IEEE, }
Jiajie Xu,~\IEEEmembership{ Member, IEEE} 
and~Mohamed-Slim~Alouini,~\IEEEmembership{Fellow,~IEEE}
\thanks{\textcolor{blue}{Qiyu Ma is a visiting student with the Communication Theory Lab at King Abdullah University of Science and Technology, Thuwal 23955, Makkah Province, Saudi Arabia (Email: qiyu.ma@kaust.edu.sa).} Jiajie Xu and Mohamed-Slim Alouini are  with the Computer, Electrical, and Mathematical Science and Engineering division in King Abdullah University of Science and Technology, Thuwal 23955, Makkah Province, Saudi Arabia. (Emails: jiajie.xu.1@kaust.edu.sa; slim.alouini@kaust.edu.sa). \textcolor{blue}{The source code for reproducing the simulation and plotting results is available at: \url{https://github.com/Yuer-1218/SA-OPS}.}}
}

\markboth{}%
{Ma \MakeLowercase{\textit{et al.}}: Integrated Discovery and State-Aware Servicing for Mobile AUVs With UWOC}

\maketitle

\begin{abstract}
Underwater wireless optical communication (UWOC) is an enabling technology for high-throughput subsea networks, yet its long-term deployment is constrained by the finite energy budget of underwater nodes. To address this challenge, we investigate a mobile system wherein an autonomous underwater vehicle (AUV) performs joint wireless information transfer (WIT) and wireless power transfer (WPT) for a network of randomly distributed sensor nodes. This paper develops \textcolor{blue}{an integrated mission-level framework} that combines stochastic node discovery with state-aware servicing. First, we present an analytical model for node discovery based on a signal-to-noise ratio (SNR) analysis, deriving performance metrics that include the probability distribution of the discovery distance. Second, we introduce \textcolor{blue}{a threshold-based scheduling framework}, termed State-Aware Optimal Point Servicing (SA-OPS), which \textcolor{blue}{selects one of three actions according to the node's real-time energy state: preemptive charging, communication followed by charging, or communication only.} Simulations and multi-criteria decision analysis show that, \textcolor{blue}{under the considered assumptions and parameter ranges}, SA-OPS can improve the tradeoff between AUV energy expenditure and network-wide energy health relative to the adopted baseline strategies. The results also indicate that the selected charging threshold can be approximated by \textcolor{blue}{a simple state-dependent heuristic}, providing a practical guideline for autonomous energy replenishment in underwater networks.
\end{abstract}


\IEEEpeerreviewmaketitle

\section{Introduction}
\subsection{Background}
With the advent of the sixth-generation (6G) mobile communication technology \cite{dang2020should}, establishing a \roundtwo{global} space-air-ground-sea integrated network (SAGSIN) has emerged as a key vision for the future \cite{xu2023space}. \textcolor{blue}{Within this vision, underwater wireless networks form a technically challenging segment that supports ocean monitoring, exploration, and preservation} \cite{kao2017comprehensive}. Traditional underwater acoustic communication, constrained by high latency and limited bandwidth, \textcolor{blue}{may be unsuitable for applications requiring high-throughput data transfer}, such as ocean monitoring, resource exploration, and the \textcolor{blue}{Internet of Underwater Things} (IoUT) \cite{DOMINGO20121879}. In this context, UWOC \textcolor{blue}{can support multi-Gbps data rates in demonstrated links \cite{lu201960m,11145354}, together with low latency and high physical-layer security}. \textcolor{blue}{Recent advances in energy-efficient \roundtwo{micro-scale light-emitting diode (micro-LED)} transmitters \cite{9878216} and underwater optical modem prototypes \cite{saksvik2025sipm} further improve the practical feasibility of UWOC networks.}

However, despite the advantages of UWOC in communication rate, large-scale and long-term deployment remains limited by the energy budget of underwater nodes/sensors. Many underwater sensor or relay nodes rely on internal batteries of finite capacity, and their replacement or recharging is expensive and often impractical \cite{s18010051}. Wireless charging technologies \cite{6266681}, particularly optical wireless power transfer \cite{9852974}, provide a possible approach for in-situ energy replenishment. This motivates a joint WIT/WPT design in which the same optical front-end supports both information delivery and energy transfer. Although the protocol considered in this work adopts a time-division implementation rather than strictly simultaneous transfer, it follows this shared-hardware design principle. In this context, we study the mission efficiency of a mobile underwater system performing joint WIT and WPT by \textcolor{blue}{combining stochastic node discovery with a state-aware servicing policy}, termed state-aware optimal point servicing (SA-OPS), to balance network sustainability and AUV operational efficiency.

\subsection{Related Work}
\textcolor{blue}{Acoustic and ultrasonic techniques can provide noncontact underwater power transfer, but practical designs remain constrained by transfer range, transducer directionality, and conversion efficiency \cite{6266681,9217956,10460553}; acoustic \roundtwo{simultaneous wireless information and power transfer (SWIPT)} also requires explicit allocation of communication and energy-transfer resources \cite{11078658}. Magnetic-induction and \roundtwo{radio-frequency (RF)} alternatives generally require short transfer distances and tight coil or antenna alignment, which complicate mobile AUV servicing \cite{mohsan2020review,mohsan2022enabling,10460553}.} Consequently, research on systems combining optical information and power transfer for underwater applications has progressed along \textcolor{blue}{several} complementary fronts: (i) theoretical performance modeling; (ii) practical system optimization\textcolor{blue}{; (iii) mobile AUV mission planning and energy-aware service control}. The existing research can be reviewed as follows.

On the theoretical front, foundational studies have established performance bounds for point-to-point dual-function links, primarily analyzing the impact of oceanic turbulence under fixed link geometries, often under a SWIPT model \cite{Ghasvarian,Uysal}. \textcolor{blue}{Concurrently, practical research has focused on link-level component optimization, including transmitter modulation that increases charging speed over established connections \cite{Maj}, wide-angle transceivers that mitigate local alignment jitter \cite{liu}, and receiver architectures that manage the tradeoff between communication and harvesting \cite{lim}.}

\textcolor{blue}{Regarding AUV mission planning, deep-reinforcement-learning and heuristic methods optimize visitation or information-collection paths over known candidate targets \cite{10582446,7932121}, while the underwater power-and-data-transfer literature generally organizes mobile servicing around planned link or rendezvous opportunities \cite{10460553}. These formulations provide useful route-planning benchmarks but do not model the stochastic optical discovery stage considered here.}
\textcolor{blue}{A related line of work studies energy-aware networking and mobile service control, including wake-up or energy-saving mechanisms for underwater networks \cite{11077794}, energy-aware path planning for AUV-assisted underwater sensor networks \cite{Acarer_2024}, and topology-adaptive or SWIPT-enabled energy management in wireless sensor networks \cite{10599885}. These studies provide useful design principles for energy-constrained service systems, but they usually focus on routing, clustering, path planning, or link-level energy transfer rather than on the coupled process of stochastic optical discovery followed by residual-energy-aware WIT/WPT servicing.}
\textcolor{blue}{Beyond underwater-specific studies, stochastic network control provides a general basis for state-dependent decisions in time-varying wireless systems \cite{neely2006energy}, while mobile charging studies have developed energy-urgency prioritization \cite{lei2023urgency} and charging-utility maximization \cite{ma2018charging}. These works establish energy-aware scheduling and adaptive replenishment as important components of energy-constrained service systems. The contribution of this study lies in integrating these components with SNR-limited stochastic optical discovery and immediate residual-energy-aware WIT/WPT servicing in a tractable AUV mission model.}

Overall, many existing studies start from a pre-established link or a deterministic service setting, which is appropriate for link-level analysis but does not fully capture the mission process of a mobile underwater system performing joint WIT and WPT. In particular, the stochastic discovery phase, in which a mobile platform must first locate a target node in a random field, is often simplified or omitted. In addition, the service decision is commonly decoupled from the node's real-time residual energy. \textcolor{blue}{These observations motivate a framework that connects random discovery with energy-state-dependent servicing while keeping the decision logic sufficiently simple for AUV operation.}

\subsection{Contribution}
Motivated by these gaps, this paper develops an integrated framework for analyzing a mobile underwater system performing joint WIT and WPT. We examine how stochastic discovery, physical-layer configurations, and operational scheduling affect mission efficiency and network sustainability. The resulting SA-OPS policy is \textcolor{blue}{a state-aware, threshold-based servicing strategy} that balances AUV energy expenditure with the residual-energy condition of the underwater network. \textcolor{blue}{Specifically, the framework addresses three coupled problems. First, it characterizes how a mobile AUV discovers a randomly located underwater node under SNR-limited optical scanning, which provides discovery-distance and search-time metrics instead of assuming a pre-established service link. Second, it connects the discovered link condition with the node's residual energy state, allowing the service action to distinguish critically depleted, intermediate-energy, and sufficiently healthy nodes. Third, it evaluates the resulting energy-service tradeoff through mission-level \roundtwo{key performance indicators (KPIs)}, so that the charging threshold can be selected according to both network sustainability and AUV energy expenditure.} The primary contributions of this paper are threefold:
\begin{itemize}

    \item \textbf{Integrated Mission Modeling.} We establish an analytical framework that \textcolor{blue}{jointly considers stochastic node discovery and subsequent mobile servicing} in a joint WIT/WPT mission. Departing from static-link-only analyses, the model incorporates a three-dimensional (3D) Poisson point process (PPP) for the network topology, enabling characterization of the mission process from initial search to service execution.
    \item \textbf{SNR-Based Performance Derivation.} \textcolor{blue}{We develop a physically grounded discovery model based on a detailed SNR analysis for silicon photomultiplier (SiPM)-based receivers, from which we derive closed-form or semi-analytical expressions for the maximum angle-dependent detection distance, effective scan success volume, expected search time, and distance distribution of the first discovered node.}
    \item \textbf{State-Aware Policy and Threshold Selection.} We propose SA-OPS as \textcolor{blue}{a state-aware threshold policy} that adapts its operational logic based on the real-time energy state of the target node. Through numerical analysis and multi-criteria decision analysis (MCDA), we examine: (i) how the policy behaves under changes in node density \textcolor{blue}{within the considered parameter range}; (ii) how the healthy-energy threshold can be formulated as a multi-objective selection problem; and (iii) how \textcolor{blue}{a simple state-dependent heuristic} can provide a practical threshold-selection guideline.

\end{itemize}

\subsection{Organization}
The remainder of this paper is organized as follows. Section \ref{sec:system_model} details the system architecture, including the network topology, the underwater optical channel model, the receiver SNR model, and the time-division operational protocol. Section \ref{sec:discovery_analysis_SNR} presents the SNR-based node discovery analysis, deriving metrics such as the success volume, expected search time, and distance distribution of the discovered node. In Section \ref{sec:optimal_scheduling}, we introduce the SA-OPS policy, \textcolor{blue}{a threshold-based scheduling rule} that integrates physical-layer power calculation with state-dependent decision logic. Section \ref{sec:simulation_results} presents Monte Carlo simulation results used to validate the analytical expressions, compare SA-OPS with the adopted baseline policies, and examine heuristic threshold selection. Finally, Section \ref{sec:conclusion} concludes the paper.

\section{System Model}
\label{sec:system_model}
In this section, we establish the foundational models for the underwater environment, communication channel, and operational protocols that underpin our analysis. \roundtwo{The main symbols and their definitions are summarized in Table \ref{tab:key_symbols}.}
\begin{table*}[!t]
\centering
\caption{Key Symbols and Definitions}\label{tab:key_symbols}
\footnotesize
\begin{tabularx}{\textwidth}{@{} c X c X @{}}
\toprule
\textbf{Symbol} & \textbf{Definition} & \textbf{Symbol} & \textbf{Definition} \\
\midrule
\multicolumn{4}{l}{\textit{\textbf{Network \& Mission Parameters}}} \\
$\lambda_{\text{node}}$ & 3D-PPP node density [\si{\per\cubic\meter}] & $d,R$ & Link distance and first-discovery distance [\si{\meter}] \\
$v,t_{\mathrm{com}}$ & AUV velocity [\si{\meter\per\second}] and payload duration [\si{\second}] & $t,\alpha(t)$ & Time and WIT/WPT switching function \\
\textcolor{blue}{$T_{\text{dwell}}$}, $N$ & \textcolor{blue}{\roundtwo{Scan duration and full-sphere scan count}} & \textcolor{blue}{$d_{\text{min}}$} & \textcolor{blue}{Closest feasible service distance} [\si{\meter}] \\
\textcolor{blue}{$W_{\text{AUV,Storage}}$} & \textcolor{blue}{Onboard AUV energy storage} [\si{\joule}] & \textcolor{blue}{$P_{\text{AUV,Platform}}$} & \textcolor{blue}{Aggregate AUV platform power} [\si{\watt}] \\
\midrule
\multicolumn{4}{l}{\textit{\textbf{Optical Link \& Environment}}} \\
$P_{Tx}(t)$ & Transmitted optical power [\si{\watt}] & \textcolor{blue}{$P_{\mathrm{Tx,com}}^*,P_{\mathrm{Tx,WPT}}$} & \textcolor{blue}{\roundtwo{Minimum WIT and rated WPT powers}} [\si{\watt}] \\
\textcolor{blue}{$m,m(t),m_{\max}$} & \textcolor{blue}{\roundtwo{Lambertian orders: static, service-stage, maximum}} & $\phi_{1/2}$ & Tx half-power semi-angle [\si{\radian}] \\
\textcolor{blue}{$\theta,\psi,\theta(t)$} & \textcolor{blue}{\roundtwo{Irradiance, incidence, and pointing-error angles}} & \textcolor{blue}{$\theta_x,\theta_y$} & \textcolor{blue}{Tangent-plane pointing errors} \\
\textcolor{blue}{$\sigma_x^2,\sigma_y^2$} & \textcolor{blue}{Pointing-jitter variances} [\si{\radian\squared}] & $H_{\text{channel}},H_{\text{move}}(t)$ & Static and service-stage channel gains \\
\textcolor{blue}{$H_{\mathrm{fad}}$} & \textcolor{blue}{\roundtwo{Channel gain with irradiance fading}} & \textcolor{blue}{$I_f,\bar I_f$} & \textcolor{blue}{Fading factor and its mean} \\
$c(\lambda),\lambda$ & Attenuation coefficient [\si{\per\meter}] and wavelength & $A_{rx},\phi_{\text{FoV}}$ & Aperture area and receiver FoV semi-angle \\
$n,g(\psi)$ & Concentrator refractive index and gain & $T_s(\psi),\mathfrak{R}(\psi)$ & Filter transmittance and receiver angular gain \\
$P_{Rx}(t),P_{\text{solar}}$ & Signal and background optical powers [\si{\watt}] & $E_{\text{sun}}(0),L_{\text{deep}}$ & Sea-level irradiance and operating depth \\
$\epsilon,\zeta_r$ & Solar attenuation and reflectance factors & $L_f,\Delta\lambda$ & Directional factor and optical-filter bandwidth \\
\midrule
\multicolumn{4}{l}{\textit{\textbf{Receiver, Noise \& Energy Parameters}}} \\
\textcolor{blue}{$G,\mathcal{R}_p,\eta$} & \textcolor{blue}{SiPM gain, primary responsivity, and photon-detection efficiency} & \textcolor{blue}{$F,B$} & \textcolor{blue}{Excess-noise factor and receiver bandwidth} \\
\textcolor{blue}{$I_d,R_L$} & \textcolor{blue}{Primary dark current and load resistance} & \textcolor{blue}{$q,h,\nu,c_0,k_B,T$} & \textcolor{blue}{\roundtwo{Electron charge, Planck constant, frequency, light speed, Boltzmann constant, temperature}} \\
SNR & Electrical signal-to-noise ratio & $\mathrm{SNR}_{\text{th,disc}},\textcolor{blue}{\mathrm{SNR}_{\text{th,com}}}$ & Discovery \textcolor{blue}{and communication} SNR threshold\textcolor{blue}{s} \\
$P_{\text{th,disc}},P_{\text{th,com}}$ & Discovery and WIT power thresholds [\si{\watt}] & $P_{\text{th,charge}}$ & \textcolor{blue}{Charging-circuit activation threshold} [\si{\watt}] \\
\roundtwo{$\eta_{\text{EH}}$} & \roundtwo{Energy-conversion efficiency} & \textcolor{blue}{$W_{\text{charge}}$} & \textcolor{blue}{Harvested energy} [\si{\joule}] \\
$W_{\text{node}},E_{\text{res}}$ & Node capacity and residual energy [\si{\joule}] & \textcolor{blue}{$P_{\text{com}},P_{\text{sleep}}$} & \textcolor{blue}{Node communication and sleep powers} [\si{\watt}] \\
\midrule
\multicolumn{4}{l}{\textit{\textbf{Policy, Optimization \& Discovery}}} \\
$E_{\text{comm}},E_{\text{healthy}}$ & Communication and healthy-energy thresholds [\si{\joule}] & \textcolor{blue}{$e,g_e$} & \textcolor{blue}{Normalized threshold and SA-OPS rule} \\
\textcolor{blue}{$\mathbf{k},\bar{\mathbf{k}},\widehat{\mathbf{k}}$} & \roundtwo{KPI vectors: realization, ensemble, sample average} & \roundtwo{$\mathcal{U}$} & \roundtwo{Multi-criteria utility} \\
$\mu,\sigma,\beta$ & Initial-energy statistics and heuristic coefficient & \textcolor{blue}{$\boldsymbol{\xi}$} & \textcolor{blue}{Finite-mission random realization} \\
$d_{\max}(\theta),V_{\text{success}}$ & Detection range and success volume & $p_s,E[T_{\text{search}}]$ & \roundtwo{Scan success probability and mean search time} \\
$V(r),F_R(r),f_R(r)$ & \textcolor{blue}{Detectable volume and \roundtwo{conditional distance distributions}} & \roundtwo{$N_{\mathrm{succ}}$} & \roundtwo{Nodes in a success volume} \\
\bottomrule
\end{tabularx}
\end{table*}
\begin{figure}[!t]
    \centering
    \includegraphics[width=\columnwidth]{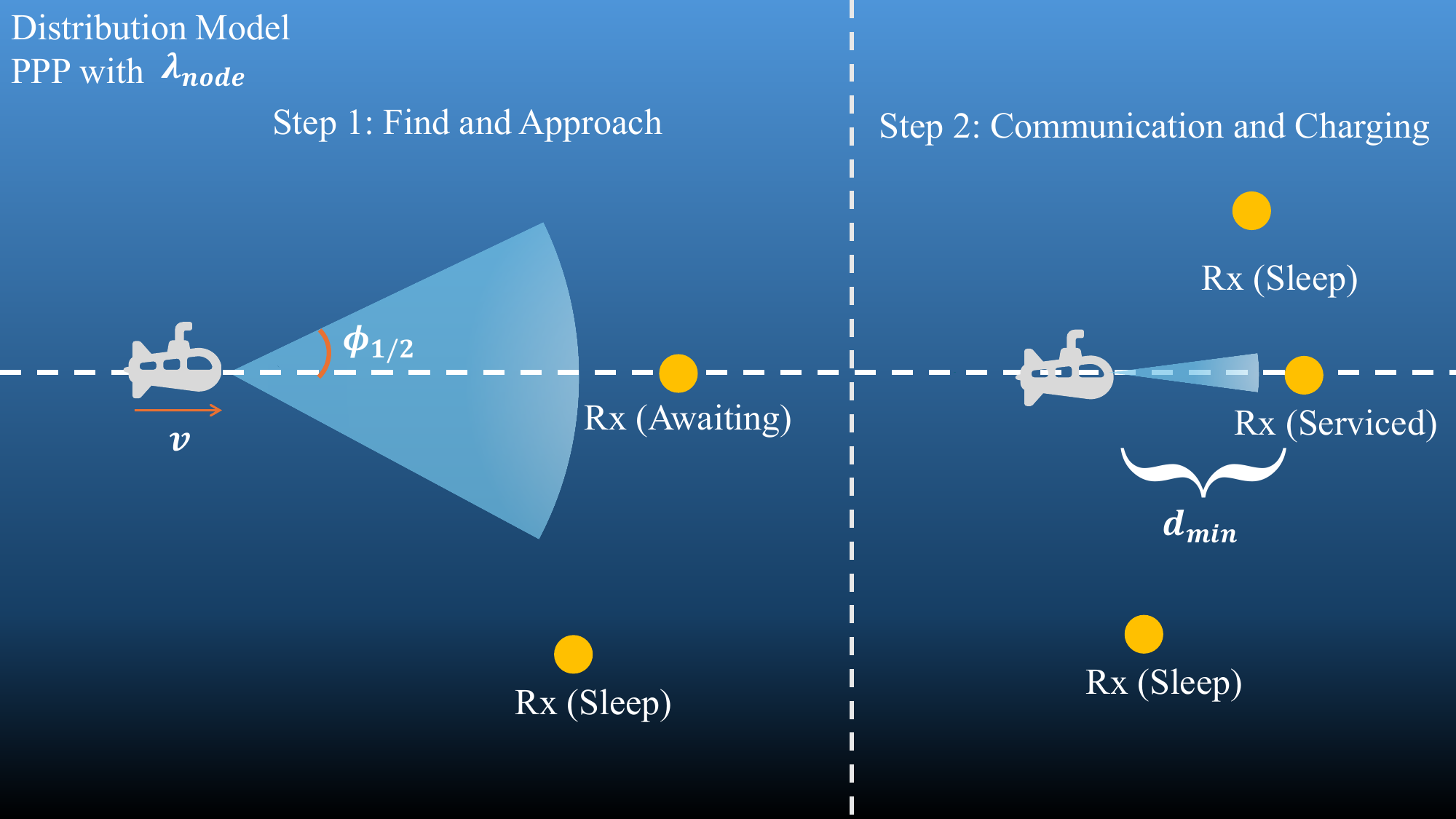}
    \caption{\textcolor{blue}{System model and mission workflow for AUV-assisted joint WIT/WPT. The AUV scans a random 3D underwater node field, approaches a node after SNR-based discovery, establishes a local state-exchange link near $d_{\text{min}}$, and then performs the required WIT/WPT service.}}
    \label{fig:system_model_illustration} 
\end{figure}

\subsection{Network Topology and Channel Model}
\textcolor{blue}{This subsection defines the spatial distribution of the underwater nodes and the optical channel used throughout the discovery and servicing analysis. We first specify the PPP-based network topology and then introduce the channel-gain expression that links geometry, attenuation, and receiver orientation.}
We consider a mobile transmitter (Tx), mounted on an AUV (underwater mobile platform), tasked with discovering and servicing a network of sensor nodes. The geometric configuration and operational phases of this system are illustrated in Fig. \ref{fig:system_model_illustration}. These receiver nodes (Rx) are assumed to be spatially distributed according to a homogeneous PPP with a constant density $\lambda_{\text{node}}$ within a 3D volume \textcolor{blue}{\cite{streit2010poisson}}. This model is widely adopted to capture random and unstructured node \textcolor{blue}{locations} \cite{elsawy2016modeling}. \textcolor{blue}{The PPP represents the spatial node field at each discovery or sampling instant and supports ensemble-averaged performance characterization of an uncoordinated deployment. Slow current-induced drift may change node locations between successive snapshots. During one local approach and handshake, the associated node is treated as quasi-static. If $v_{\text{drift}}$ is the relative drift speed and $T_{\text{loc}}$ is the approach-plus-handshake duration, this approximation requires $v_{\text{drift}}T_{\text{loc}}$ to remain small relative to the tracking margin and receiver field-of-view footprint. The periodic tracking pings in Phase 1 update the line-of-sight direction within this interval; larger displacements require reacquisition or an explicit mobility filter.}

\textcolor{blue}{A scan can contain more than one detectable node. The discovery layer therefore associates the AUV with the node having the largest estimated received optical power; under comparable orientation and attenuation conditions, this rule favors the nearest node. Once selected, that node remains exclusively associated with the AUV through approach, state exchange, and service completion, while the other detected nodes are deferred to later scans. This rule defines the single-node service abstraction used in the analysis; contention among simultaneous responses and joint multi-node scheduling remain outside the present protocol model.}

\textcolor{blue}{Fig.~\ref{fig:system_model_illustration} summarizes the mission workflow considered in this paper. The AUV first scans the random node field and declares a node discoverable when the received SNR exceeds the discovery threshold. It then approaches the associated node, establishes a short local state-exchange link near $d_{\text{min}}$, and performs WIT/WPT service. The state-dependent service rule itself is introduced later in Section~\ref{sec:optimal_scheduling}.}

The channel gain of the UWOC link, which dictates the power transfer efficiency between the Tx and an Rx, is a function of the distance $d$, the transmitter's irradiance angle $\theta$, and the receiver's incidence angle $\psi$. Following a widely adopted and validated model \cite{kahn1997wireless,hamza2016investigation}, the channel gain is expressed as 
\begin{align}
    H_{\text{channel}}(d, \theta, \psi) = \frac{A_{rx}(m+1)}{2\pi d^2} \cos^m(\theta) e^{-c(\lambda)d} \mathfrak{R}(\psi),
    \label{eq:channel_gain_general_model}
\end{align}
where $m = -\ln(2)/\ln(\cos(\phi_{1/2}))$ is the Lambertian order related to the transmitter's semi-angle at half power, $\phi_{1/2}$. The term $c(\lambda)$ represents the beam attenuation coefficient for a given wavelength $\lambda$, and $A_{rx}$ is the receiver's physical aperture area. The composite receiver-side angular gain, $\mathfrak{R}(\psi)$, is defined as 
\begin{align}
    \mathfrak{R}(\psi) = T_s(\psi) g(\psi) \cos(\psi),
\end{align}
which incorporates the optical filter transmittance $T_s(\psi)$ and the gain of the non-imaging optical concentrator, $g(\psi)$. The concentrator gain is given by
\begin{align}
g(\psi) = 
\begin{cases}
    n^2 / \sin^2(\phi_{\text{FoV}}), & 0 \le \psi \le \phi_{\text{FoV}},  \\
    0, & \psi > \phi_{\text{FoV}},
\end{cases}
\end{align}
where $\phi_{\text{FoV}}$ is the receiver's field-of-view (FoV) semi-angle and $n$ is its internal refractive index.
\textcolor{blue}{Physically, \eqref{eq:channel_gain_general_model} separates the received-power loss into geometric spreading through $1/d^2$, transmitter directivity through $\cos^m(\theta)$, water absorption and scattering through $e^{-c(\lambda)d}$, and receiver orientation and collection efficiency through $\mathfrak{R}(\psi)$.}

{\color{blue}
The deterministic channel gain in \eqref{eq:channel_gain_general_model} represents the Beer--Lambert/Lambertian component of the link budget. If turbulence-induced scintillation is included, it can be extended by a nonnegative multiplicative irradiance-fading factor as
\begin{align}
    H_{\mathrm{fad}}(d,\theta,\psi)
    = H_{\text{channel}}(d,\theta,\psi) I_f,
    \label{eq:fading_channel_extension}
\end{align}
where $I_f$ captures random irradiance fluctuations and $\bar I_f\triangleq\mathbb{E}[I_f]$ denotes its mean. Hence, quantities that are linear in the channel gain satisfy $\mathbb{E}[H_{\mathrm{fad}}]=H_{\text{channel}}\bar I_f$. Common choices for $I_f$ in UWOC turbulence studies include lognormal, Gamma--Gamma, and exponential-generalized-Gamma models \cite{oubei2017simple,zedini2019unified}. Hereafter, $H_{\mathrm{fad}}$ denotes the effective link gain. In the numerical evaluation, $I_f$ is retained as a single aggregate channel factor; its turbulence-, beam-wander-, biofouling-, and scattering-related components are not separately decomposed.
}

\subsection{Receiver SNR Model and Performance Thresholds}
\textcolor{blue}{This subsection translates the received optical power into communication and discovery feasibility through an SNR model. The resulting power thresholds provide the common interface between the physical-layer channel model and the higher-level discovery and servicing decisions.}
For reliable operation, the received optical power must be sufficient to meet distinct performance criteria for \textcolor{blue}{initial node discovery, subsequent data communication, and effective energy charging}. We ground these criteria in the electrical SNR at the receiver, which provides a more physically meaningful basis than simple power thresholds.

For a SiPM-based receiver, the electrical SNR \textcolor{blue}{determines the adopted detection and communication criteria}. The total receiver noise \textcolor{blue}{combines several independent sources}. The signal photocurrent after avalanche gain is given by
\begin{align}
    I_{\text{signal}} = G \mathcal{R}_p P_{Rx},
\end{align}
where $G$ is the SiPM gain, $P_{Rx}$ is the received optical signal power, and $\mathcal{R}_p = \eta q / (h\nu) = \eta \lambda q / (h\textcolor{blue}{c_0})$ is the primary responsivity of the photodetector. \textcolor{blue}{Here, $q$ is the electron charge, $h$ is Planck's constant, $\nu$ is the optical frequency, and $c_0$ is the speed of light in vacuum.}

The total noise variance, $\sigma_{\text{total}}^2$, is the sum of shot noise and thermal noise,
\begin{align}
    \sigma_{\text{total}}^2 = \sigma_{\text{shot}}^2 + \sigma_{\text{th}}^2.
\end{align}
The shot noise arises from statistical fluctuations in the total primary current, which includes contributions from the signal, background solar radiation, and intrinsic dark current \cite{ghassemlooy2019optical}. Its variance is
\begin{align}
    \sigma_{\text{shot}}^2 = 2q G^2 F B \left( \mathcal{R}_p P_{Rx} + \mathcal{R}_p P_{\text{solar}} + \textcolor{blue}{I_d} \right),
\end{align}
where $F$ is the excess noise factor, $B$ is the receiver bandwidth, $I_{d}$ is the primary dark current (before gain), and $P_{\text{solar}}$ is the incident background solar power. The incident background solar power, $P_{\text{solar}}$, is modeled as \cite{1605919, 4686801},
\begin{align}
    P_{\text{solar}} = A_{rx} \phi_{\text{FoV}}^2 L_f \zeta_r E_{\text{sun}}(z=0) e^{-\epsilon L_{\text{deep}}} \Delta\lambda T_s(\psi).
    \label{eq:p_solar_definition}
\end{align}
The thermal noise variance is independent of the optical signal and is given by
\begin{align}
    \sigma_{\text{th}}^2 = \frac{4k_B T B}{R_L},
\end{align}
where $k_B = 1.38\times10^{-23}  J/K$ is the Boltzmann constant, $T$ is the absolute temperature, and $R_L$ is the load resistance.

The overall electrical SNR, defined as the ratio of the squared signal current to the total noise variance, is thus given by the comprehensive model in \eqref{eq:snr_full} (presented at the bottom of this page). \textcolor{blue}{The numerator represents the useful electrical signal power generated from the received optical signal, whereas the denominator aggregates the random fluctuations caused by signal-dependent shot noise, background-light noise, dark current, and receiver thermal noise. An algebraically equivalent SiPM SNR construction, combining the same signal and receiver-noise components, is used in \cite{ma2026offset}.}

\begin{figure*}[b]
\normalsize
\hrulefill
\begin{equation} \label{eq:snr_full}
\mathrm{SNR}(P_{Rx}) = \frac{\left(G \mathcal{R}_p P_{Rx}\right)^2}{2qG^2F B \left(\mathcal{R}_p P_{Rx} + \mathcal{R}_p P_{\text{solar}} + \textcolor{blue}{I_d}\right) + \frac{4k_B T B}{R_L}}.
\end{equation}
\end{figure*}

\textcolor{blue}{This SNR model leads to three task-dependent received-power thresholds with different physical roles. The discovery threshold $P_{\text{th,disc}}$ is the minimum received optical power required for the receiver signal to be distinguished from noise during the scanning stage; it defines whether a node belongs to the detectable region used in the stochastic discovery analysis. The communication threshold $P_{\text{th,com}}$ is the received optical power adopted to support reliable payload WIT; in the benchmark link budget, it is set to $P_{\text{th,com}}=\SI{1.2}{\micro\watt}$. This threshold is higher than $P_{\text{th,disc}}$ because payload communication imposes a more stringent link requirement than initial signal detection. The charging threshold $P_{\text{th,charge}}$ is the minimum received input power at which the harvesting circuit is treated as active; it does not by itself guarantee positive net battery replenishment, which additionally requires the harvested \roundtwo{direct-current (DC)} power to exceed the node load. Together, these thresholds connect the optical channel and receiver model to the three mission operations considered in this paper: discovering a node, communicating with it, and initiating energy harvesting.}

\subsection{Mobile Operation and Time-Division Protocol}
\textcolor{blue}{This subsection describes how the AUV uses the optical link over time after a node has been discovered. We define the time-division operation between WIT and WPT and use it to express the node energy dynamics that later drive the state-aware policy.}
After a node is discovered at an initial distance $R$, the AUV moves towards it at a constant velocity $v$. \roundtwo{During this approach, the separation is $d(\tau)=R-v\tau$, where $\tau$ denotes elapsed travel time.}

Pointing errors arise from the AUV's jitter. \textcolor{blue}{The independent horizontal and vertical tangent-plane components are modeled as $\theta_x\sim\mathcal{N}(0,\sigma_x^2)$ and $\theta_y\sim\mathcal{N}(0,\sigma_y^2)$. The radial pointing-error angle $\theta(t)$, measured from the nominal line of sight, satisfies $\tan\theta=\sqrt{\theta_x^2+\theta_y^2}$. The resulting general angular \roundtwo{probability density function (PDF)} is the Hoyt/Nakagami-$q$ form \cite{simon2004digital,paris2009nakagami}}
{\color{blue}
\begin{equation}
\begin{aligned}
p(\theta)
&=\frac{\tan\theta\sec^2\theta}{\sigma_x\sigma_y}
\exp\!\left[-\frac{\tan^2\theta}{4}
\left(\frac{1}{\sigma_x^2}+\frac{1}{\sigma_y^2}\right)\right]\\
&\quad\times I_0\!\left[\frac{\tan^2\theta}{4}
\left(\frac{1}{\sigma_x^2}-\frac{1}{\sigma_y^2}\right)\right],
\quad \theta\in[0,\pi/2),
\end{aligned}
\label{eq:exact_angle_pdf}
\end{equation}
}
\textcolor{blue}{where $I_0(\cdot)$ is the modified Bessel function of the first kind. When $\sigma_x=\sigma_y=\sigma$, $I_0(0)=1$ and \eqref{eq:exact_angle_pdf} reduces to the Rayleigh-based angular PDF \cite{4267802}.}
Due to this random jitter, the instantaneous received power is a time-varying random process:
\begin{equation}
P_{\mathrm{Rx}}(t) = P_{\mathrm{Tx}}(t) H_{\mathrm{move}}(t),
\end{equation}
where $P_{\mathrm{Tx}}(t)$ is the transmitted power and $H_{\mathrm{move}}(t)$ is the time-varying channel gain, which depends on the random angle $\theta(t)$.

To accommodate both WIT and WPT, the time-division multiplexing protocol is applied, which is central to the scheduling policies analyzed in Section \ref{sec:optimal_scheduling}. We define a binary switching function $\alpha(t)$ as
\begin{equation}
\alpha(t) =\begin{cases}
    1,  & \text{for WIT at time } t,\\
    0,  & \text{for WPT at time } t.
\end{cases}
\end{equation}
The total energy harvested by the receiver \roundtwo{over the service interval} of duration $t_{\mathrm{end}}$ is therefore \textcolor{blue}{given by}\footnote{\roundtwo{Equation~\eqref{eq:energy_harvested} is a general expression for harvested energy over an interval of WIT/WPT operation. Under the service policy introduced in Section~\ref{sec:optimal_scheduling}, charging takes place while the AUV hovers at $d_{\min}$, after completing its approach; no WPT is performed during discovery or transit. The integral is evaluated over the service interval with $d(t)=d_{\min}$, while pointing jitter can still vary the instantaneous received power.}}
\begin{equation}
  W_{\mathrm{charge}}=\eta_{\text{EH}}\int_0^{t_{\mathrm{end}}} P_{\mathrm{Rx}}(t) (1-\alpha(t)) \mathrm{d}t.
  \label{eq:energy_harvested}
\end{equation}
\textcolor{blue}{Communication and charging use different power rules. During WIT, the AUV selects the minimum transmit power that satisfies the communication SNR requirement. During WPT, it uses the rated charging power $P_{\mathrm{Tx,WPT}}$ subject to the transmitter limit, while $P_{\text{th,charge}}$ determines whether the harvesting circuit is active. The constant $\eta_{\text{EH}}$ in \eqref{eq:energy_harvested} is an effective mission-level conversion efficiency. Underwater photovoltaic conversion depends on the device and operating conditions \cite{rohr2023dive}, while practical \roundtwo{simultaneous lightwave information and power transfer (SLIPT)} prototypes exhibit hardware-dependent charging behavior \cite{Maj,lim}. With circuit-specific measurements, the linear conversion term $\eta_{\mathrm{EH}}P_{\mathrm{Rx}}$ can be replaced by a calibrated harvested DC-power function $f_{\mathrm{EH}}(P_{\mathrm{Rx}})$, giving}
{\color{blue}
\begin{equation}
W_{\mathrm{charge}}^{\mathrm{nl}}
=\int_0^{t_{\mathrm{end}}}f_{\mathrm{EH}}\!\left(P_{\mathrm{Rx}}(t)\right)
(1-\alpha(t))\,\mathrm{d}t.
\label{eq:nonlinear_harvested_energy}
\end{equation}
}
\textcolor{blue}{The node's net battery change is obtained separately by subtracting its state-dependent load and applying the physical battery bounds. Under the adopted aggregate link-budget model, the received power at the 1-m service point remains many orders of magnitude above the 400-nW activation threshold. This comparison is used only to establish that the considered benchmark does not operate in the ultra-low-input turn-on region; the absolute charging power remains subject to finite-aperture overlap, pointing loss, saturation, and circuit-specific calibration. These effects can change the charging rate and thereby shift the service time, AUV energy expenditure, and calibrated value of $\beta$, and their circuit-level characterization is left for future work.}
\textcolor{blue}{These equations provide the communication and harvesting energy accounting used by the subsequent SA-OPS policy.}

\subsection{\textcolor{blue}{Problem Formulation}}
\label{subsec:problem_formulation}
\textcolor{blue}{The mission-level problem couples random node discovery, local service execution, and finite-energy operation. Given a random underwater node field, an optical channel model, node energy states, and a finite AUV energy budget, the AUV repeatedly discovers a node, moves to a feasible service location, and allocates optical transmission time and power between WIT and WPT. Each service must satisfy the SNR requirements for discovery and communication, the charging-circuit activation and net-replenishment conditions, the transmit-power limit, and the AUV energy budget.}

\textcolor{blue}{The decision family is the SA-OPS rule $g_e$, indexed by the normalized healthy-energy threshold $e=E_{\text{healthy}}/W_{\text{node}}$. Let $\boldsymbol{\xi}$ denote one finite-mission realization of the random node field and initial energy states, and let $\mathbf{k}(g_e;\boldsymbol{\xi})$ collect the resulting mission-level KPIs, including network survival time, critical-node rescue efficiency, AUV energy consumption, delivered energy, and final network-energy variance. The ensemble-averaged KPI vector is $\bar{\mathbf{k}}(g_e)=\mathbb{E}_{\boldsymbol{\xi}}[\mathbf{k}(g_e;\boldsymbol{\xi})]$, and the offline threshold-selection problem is written as}
\begin{align}
    \textcolor{blue}{e^* = \arg\max_{e \in \mathcal{E}} \mathcal{U}\!\left(\bar{\mathbf{k}}(g_e)\right)}.
    \label{eq:problem_formulation}
\end{align}
\textcolor{blue}{Here, $\mathcal{E}$ is the feasible threshold set and $\mathcal{U}(\cdot)$ is a mission-dependent multi-criteria utility. In the numerical study, $\bar{\mathbf{k}}(g_e)$ is approximated by a finite-sample Monte Carlo average and $e^*$ is selected offline using MCDA. The resulting architecture combines stochastic discovery, local threshold-based service, ensemble-averaged finite-mission performance analysis, and sample-average threshold calibration. Once $e^*$ is calibrated, online execution requires only the residual-energy report of the associated node and a constant-complexity branch selection.}
\textcolor{blue}{This formulation covers the chain from stochastic node detection and association to local state-aware servicing. It does not solve the complementary global mission-planning problem with multi-target trajectory optimization, current-aware navigation, or decomposed propulsion dynamics; these layers can be integrated with the present structure in future work.}
\textcolor{blue}{The analysis is based on controlled assumptions. First, the numerical evaluation represents irradiance variation through the aggregate factor $I_f$ in \eqref{eq:fading_channel_extension}, without separately resolving turbulence, beam wander, biofouling, or time-varying scattering. Second, the associated node is quasi-static during one local approach-and-service interval. Third, residual-energy feedback is available after the acknowledged state exchange in Section~\ref{sec:optimal_scheduling}. Fourth, AUV propulsion is represented by an aggregate platform-power budget rather than a decomposed hydrodynamic model. These assumptions bound the results to controlled local service scenarios; distribution-specific fading, localization errors, current-aware mobility, biofouling, and detailed propulsion remain future work.}

\section{SNR-Based Node Discovery Analysis}
\label{sec:discovery_analysis_SNR}
\textcolor{blue}{This section develops an SNR-based analysis of the node discovery process. A node is declared discovered when its instantaneous electrical SNR exceeds the threshold $\mathrm{SNR}_{\text{th,disc}}$; as specified in Section \ref{sec:system_model}, the numerical analysis uses $\mathrm{SNR}_{\text{th,disc}} = \SI{3}{\decibel}$.}
\textcolor{blue}{The derivation proceeds in three steps. First, the SiPM receiver model is converted into an equivalent received-power threshold by solving the SNR constraint. Second, this threshold is combined with the UWOC channel gain to define an angle-dependent detectable region and its success volume. Third, the success volume is inserted into the void probability of the 3D PPP to obtain the discovery probability, expected search time, and first-discovered-node distance distribution. This sequence separates the physical-layer link-budget calculation from the stochastic-geometry step and clarifies the assumptions used in each part of the derivation.}

\subsection{Effective Power Threshold for Discovery}
\textcolor{blue}{This subsection establishes the received-power threshold required for a node to be considered discoverable. This threshold converts the SNR requirement into a distance- and angle-dependent detection condition used in the following geometric analysis.}
The first step is to translate the SNR criterion into an equivalent minimum required optical power, which we denote $P_{\text{th,disc}}$, and this is accomplished by solving $\mathrm{SNR}(P_{Rx}) = \mathrm{SNR}_{\text{th,disc}}$ using Equation~(\ref{eq:snr_full}). This yields a quadratic equation in $P_{Rx}$
\begin{align}
\!(G\mathcal{R}_p)^2 P_{Rx}^2 
\!-\! (2qG^2FB\mathcal{R}_p \mathrm{SNR}_{\text{th,disc}}) P_{Rx}\! -\!\mathrm{SNR}_{\text{th,disc}} N_0 \!=\! 0,
\end{align}
where $N_0 = 2qG^2FB(\mathcal{R}_p P_{\text{solar}} + I_{d}) + \frac{4k_BTB}{R_L}$ represents the signal-independent noise floor. The positive root of this equation yields the exact discovery power threshold, $P_{\text{th,disc}}$. \textcolor{blue}{To expose the intermediate step, define \(a=(G\mathcal{R}_p)^2\), \(b=2qG^2FB\mathcal{R}_p\mathrm{SNR}_{\text{th,disc}}\), and \(c=\mathrm{SNR}_{\text{th,disc}}N_0\). The threshold equation is then \(aP_{Rx}^2-bP_{Rx}-c=0\), whose physically admissible root is}
\begin{align*}
\textcolor{blue}{
P_{\text{th,disc}}
=\frac{b+\sqrt{b^2+4ac}}{2a}.}
\end{align*}
\textcolor{blue}{When the signal-dependent shot-noise contribution is small at the discovery boundary, i.e., \(bP_{\text{th,disc}}\ll c\) (equivalently \(b^2\ll4ac\)), the linear term can be neglected relative to the signal-independent noise term. The exact root then reduces to}
\begin{align}
    P_{\text{th,disc}} \approx \frac{\sqrt{\mathrm{SNR}_{\text{th,disc}} N_0}}{G\mathcal{R}_p}.
    \label{eq:prx_threshold_approx}
\end{align}
This value, $P_{\text{th,disc}}$, \textcolor{blue}{is used in the subsequent discovery analysis}. A node at distance $d$ and angle $\theta$ is detectable if $P_{Tx} H_{\mathrm{fad}}(d, \theta, \psi) > P_{\text{th,disc}}$.
\textcolor{blue}{The approximation in \eqref{eq:prx_threshold_approx} shows the physical meaning of the discovery threshold: a larger noise floor or a stricter discovery SNR requirement increases the required received optical power, whereas a larger detector gain or responsivity lowers the required optical power.}

\subsection{Maximum Detection Distance and Success Volume}
\textcolor{blue}{This subsection characterizes the spatial region in which a node can be discovered during a scan. We first derive the maximum detectable distance as a function of the off-axis angle and then integrate this boundary to obtain the effective success volume.}
\begin{proposition}[Maximum Detection Distance]
\label{prop:dmax_snr}
Assuming parallel alignment ($\theta=\psi$) and a receiver within the transmitter's beam and its own FoV, the maximum range at which the receiver can be discovered at an angle $\theta$ from the transmitter's axis is given by
\begin{align}
    d_{\max}(\theta) 
    = \frac{2}{c(\lambda)}
    W_0\!\left(
    \frac{c(\lambda)}{2}
    \sqrt{K_{\mathrm{SNR}}\cos^{m+1}(\theta)}
    \right),
    \label{eq:dmax_snr_result}
\end{align}
where $W_0(\cdot)$ is the principal branch of the Lambert $W$ function and the constant $K_{\mathrm{SNR}}$ is defined as 
\begin{align}
    K_{\mathrm{SNR}} 
    = \frac{P_{Tx}A_{rx}(m+1)\textcolor{blue}{T_s(\theta)g(\theta)\bar I_f}}{2\pi P_{\text{th,disc}}}.
\end{align}
\end{proposition}

\begin{figure}[!t]
    \centering
     \includegraphics[width=\columnwidth]{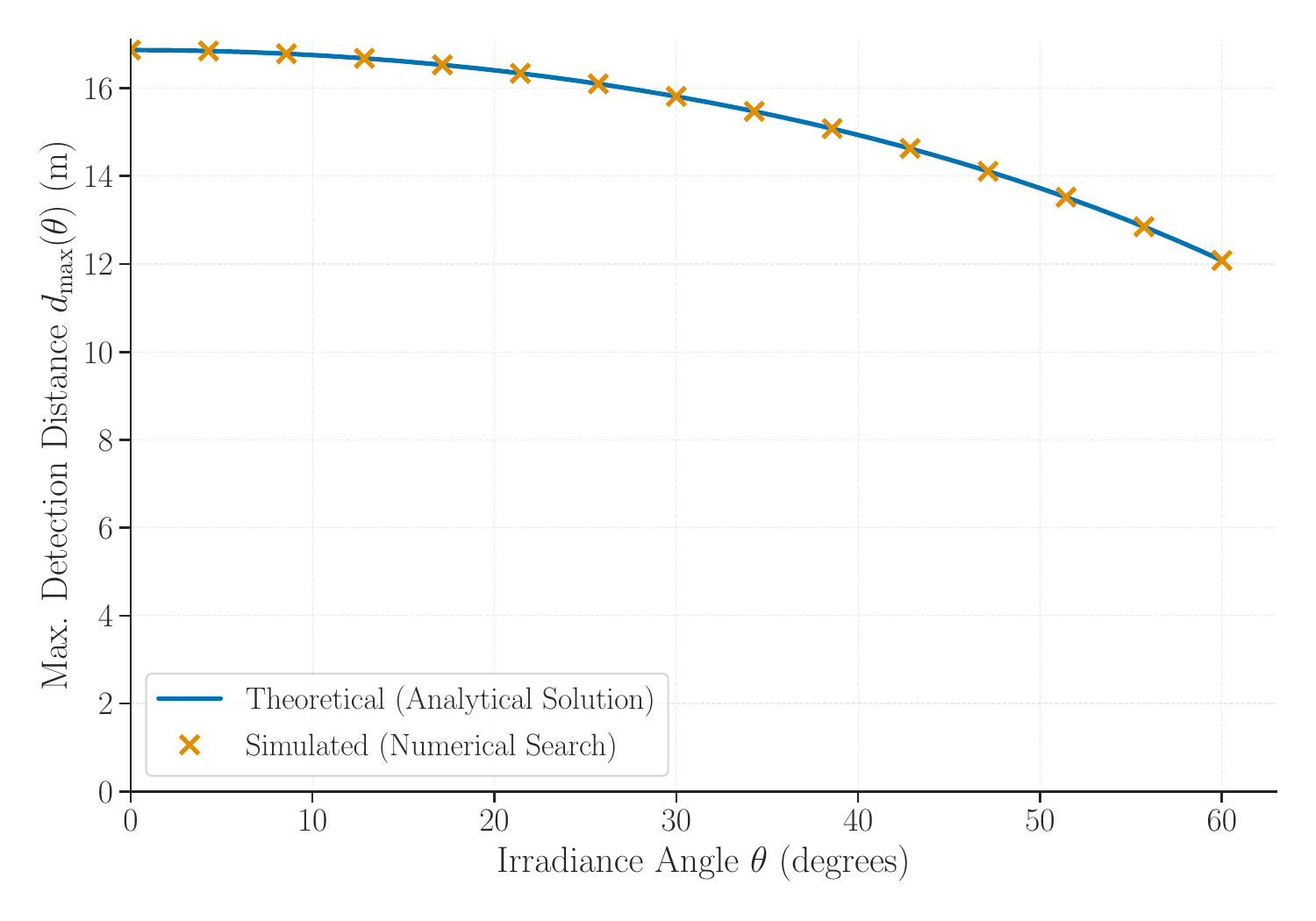}

    \caption{Maximum detection distance $d_{\max}(\theta)$ as a function of the off-axis angle $\theta$, comparing the analytical model  \eqref{eq:dmax_snr_result} with Monte Carlo simulations.}
    \label{fig:verification_d_max}
\end{figure}

\begin{proof}
\textcolor{blue}{The tractable discovery boundary is defined using the aggregate effective gain, $P_{Tx}\mathbb{E}[H_{\mathrm{fad}}(d,\theta,\psi)]=P_{Tx}H_{\mathrm{channel}}(d,\theta,\psi)\bar I_f=P_{\text{th,disc}}$. Under the parallel-alignment assumption $\psi=\theta$, the receiver angular term becomes $\mathfrak{R}(\theta)=T_s(\theta)g(\theta)\cos\theta$, which gives}
\begin{align}
\textcolor{blue}{
P_{\text{th,disc}}
 = \frac{P_{Tx}A_{rx}(m+1)T_s(\theta)g(\theta)\bar I_f}{2\pi d^2}
\cos^{m+1}(\theta)e^{-c(\lambda)d}.}
\end{align}
\textcolor{blue}{Rearranging yields $d^2e^{c(\lambda)d}=K_{\mathrm{SNR}}\cos^{m+1}(\theta)$. Since $d\ge0$ and $c(\lambda)>0$, taking the positive square root gives}
\begin{align}
d e^{\frac{1}{2}c(\lambda)d}
=\sqrt{K_{\mathrm{SNR}}\cos^{m+1}(\theta)}.
\end{align}
\textcolor{blue}{Multiplying both sides by $c(\lambda)/2$ and defining $z=c(\lambda)d/2$ gives $z e^z=\frac{c(\lambda)}{2}\sqrt{K_{\mathrm{SNR}}\cos^{m+1}(\theta)}$. Applying the principal Lambert-$W$ branch, $z=W_0(z e^z)$, yields \eqref{eq:dmax_snr_result}.}
\end{proof}
\textcolor{blue}{Thus, $d_{\max}(\theta)$ is the boundary at which the received power just equals the discovery threshold. A larger transmit power or receiver aperture expands this boundary, whereas stronger water attenuation, larger off-axis angle, or a higher discovery threshold reduces the discoverable range.}

\begin{corollary}[Effective Success Volume]
\label{cor:vsuccess_snr}
The volume in which a receiver can be discovered during a single scan (within the beam's semi-angle $\phi_{1/2}$) is found by integrating the elemental volumes over the solid angle of the beam:
\begin{align}
    V_{\text{success}}
    = \int_{0}^{2\pi} \int_{0}^{\phi_{1/2}} \int_{0}^{d_{\max}(\theta)} \rho^2 \sin(\theta) \,\mathrm{d}\rho \,\mathrm{d}\theta \,\mathrm{d}\varphi,
\end{align}
which simplifies to
\begin{align}
    V_{\text{success}}
    = \frac{2\pi}{3}
    \int_{0}^{\phi_{1/2}} 
    \left[d_{\max}(\theta)\right]^3
    \sin(\theta)\,\mathrm{d}\theta.
    \label{eq:vsuccess_snr}
\end{align}
This volume is a critical parameter, as it quantifies the reach of a single discovery attempt.\footnote{\roundtwo{During discovery, the AUV hovers at a fixed position and maintains an approximately fixed optical pointing direction within each dwell; the beam may be redirected between dwells. The speed $v$ applies to the subsequent approach to a discovered node. Throughout the search interval, the AUV consumes energy for station keeping, onboard operation, and discovery transmission. These costs are included in the mission energy budget through $P_{\text{AUV,Platform}}$ and the discovery-transmission power, as specified in Section~\ref{sec:setup_and_metrics}.}}
\textcolor{blue}{Physically, $V_{\text{success}}$ in \eqref{eq:vsuccess_snr} is the effective 3D \roundtwo{detectable volume during one fixed-pose discovery dwell}; it integrates the angle-dependent discovery boundary over the beam solid angle rather than assuming a fixed spherical range.}
\end{corollary}

\begin{figure}[!t]
    \centering

    \includegraphics[width=\columnwidth]{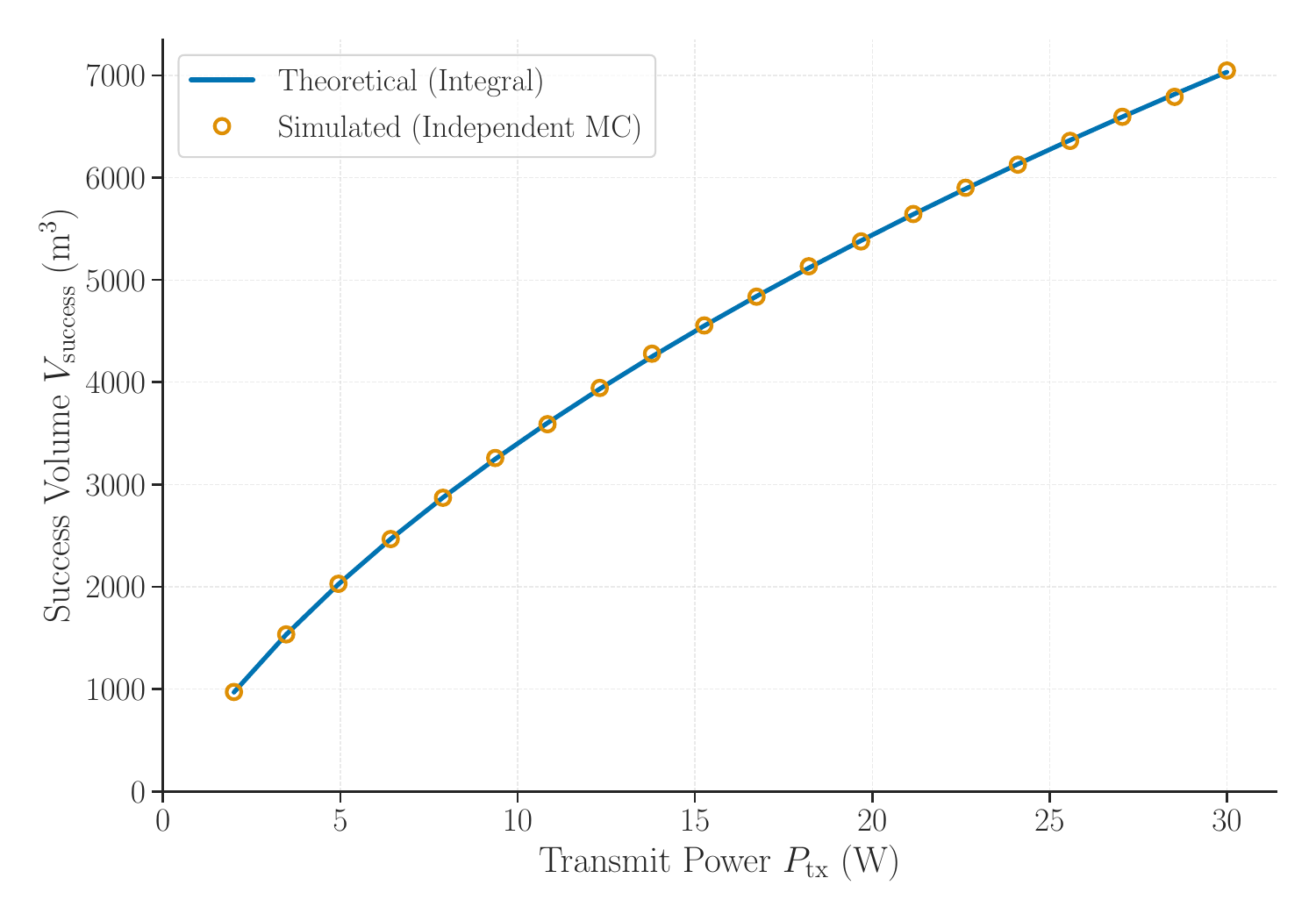}

    \caption{Effective success volume $V_{\text{success}}$ as a function of transmitter power $P_{Tx}$, validating the integral expression \eqref{eq:vsuccess_snr} against simulation results.}
    \label{fig:verification_V_success}
\end{figure}

\textcolor{blue}{Fig.~\ref{fig:verification_V_success} compares the integral expression in \eqref{eq:vsuccess_snr} with Monte Carlo simulations over different transmit-power values. The close agreement between the analytical curve and the simulated points indicates that the success-volume calculation captures the geometry of the detectable region and the SNR-based discovery boundary used in the simulations. This validation is important because \(V_{\text{success}}\) is the input to the subsequent PPP-based discovery probability and search-time analysis.}

\subsection{Discovery Probabilities and Search Metrics}
\textcolor{blue}{This subsection connects the success volume to stochastic search performance under the PPP node model. The derived metrics quantify the probability of discovering at least one node and the expected time required for the first successful discovery.}
Based on the PPP assumption and the derived success volume, we can characterize the stochastic nature of the discovery process.
\textcolor{blue}{Using the effective scan volume $V_{\text{success}}$ defined in \eqref{eq:vsuccess_snr}, the PPP void probability gives the single-scan discovery probability $p_s$ in \eqref{eq:ps_snr}: it is the complement of the probability that no node lies inside the detectable volume. Equation \eqref{eq:tsearch_snr} then converts this per-scan success probability into the expected search time under the adopted repeated-scan model, with each attempt lasting $T_{\text{dwell}}$. Finally, \eqref{eq:psuccess_snr} extends the single-scan result to $N$ non-overlapping scans by taking the complement of the joint scan-failure probability.}

\begin{theorem}[Discovery Performance Metrics]
\label{thm:discovery_perf_snr}
For a receiver density $\lambda_{\text{node}}$ and a scan duration $T_{\text{dwell}}$, \roundtwo{assuming independent equal-volume discovery opportunities,} the key discovery metrics are
\begin{enumerate}
    \item The probability of finding at least one node in a single scan is
    \begin{align}
        p_s = 1 - \exp(-\lambda_{\text{node}}V_{\text{success}}).
        \label{eq:ps_snr}
    \end{align}
    \item The expected search time to discover the first node is
    \begin{align}
        E[T_{\text{search}}]
        = \frac{T_{\text{dwell}}}{p_s}
        = \frac{T_{\text{dwell}}}{1 - e^{-\lambda_{\text{node}}V_{\text{success}}}}.
        \label{eq:tsearch_snr}
    \end{align}
    \item The overall probability of finding a node after $N$ non-overlapping scans \roundtwo{in an idealized full-sphere coverage approximation} is
    \begin{align}
        P_{\text{success}}
        = 1 - \exp(-N\lambda_{\text{node}}V_{\text{success}}),
        \label{eq:psuccess_snr}
    \end{align}
    where $N \approx 2 / (1 - \cos(\phi_{1/2}))$.
\end{enumerate}
\end{theorem}
\begin{proof}
The proof relies on fundamental properties of the PPP. The number of nodes in a volume $V$ follows a Poisson distribution with mean $\lambda_{\text{node}} V$. The probability of finding zero nodes (the void probability) in the success volume $V_{\text{success}}$ is $\exp(-\lambda_{\text{node}} V_{\text{success}})$. The probability of finding at least one node is the complement, thereby proving (\ref{eq:ps_snr}). The search for a node can be modeled as a sequence of independent Bernoulli trials, where each scan is a trial with success probability $p_s$. The number of trials until the first success follows a geometric distribution with mean $1/p_s$. As each trial has a duration of $T_{\text{dwell}}$, the expected total search time is $T_{\text{dwell}}/p_s$, proving (\ref{eq:tsearch_snr}). The probability of failing to find a node in one scan is $1-p_s = \exp(-\lambda_{\text{node}} V_{\text{success}})$. The probability of failing in all $N$ independent scans is $(1-p_s)^N = \exp(-N \lambda_{\text{node}} V_{\text{success}})$. The overall success probability is its complement, proving (\ref{eq:psuccess_snr}).
\textcolor{blue}{Equivalently, if $N_{\mathrm{succ}}$ denotes the number of nodes in one success volume, then $N_{\mathrm{succ}}\sim\mathrm{Poisson}(\lambda_{\text{node}}V_{\text{success}})$ and}
\begin{align}
\textcolor{blue}{
p_s=P(N_{\mathrm{succ}}\ge1)=1-P(N_{\mathrm{succ}}=0)
=1-e^{-\lambda_{\text{node}}V_{\text{success}}}.}
\end{align}
\textcolor{blue}{For the full-sphere scanning case, the solid angle of one conical scan is $\Omega_b=2\pi(1-\cos\phi_{1/2})$, while the sphere has solid angle $4\pi$, giving $N\approx4\pi/\Omega_b=2/(1-\cos\phi_{1/2})$.}
\end{proof}

\roundtwo{The waiting-time model uses independent discovery opportunities with equal success probability. In a static node field, disjoint detectable regions have independent PPP node counts, whereas overlapping or repeated observations are correlated. Beam-redirection and relocation overheads are not separately resolved in this benchmark.}

\begin{figure}[!t]
    \centering

    \includegraphics[width=\columnwidth]{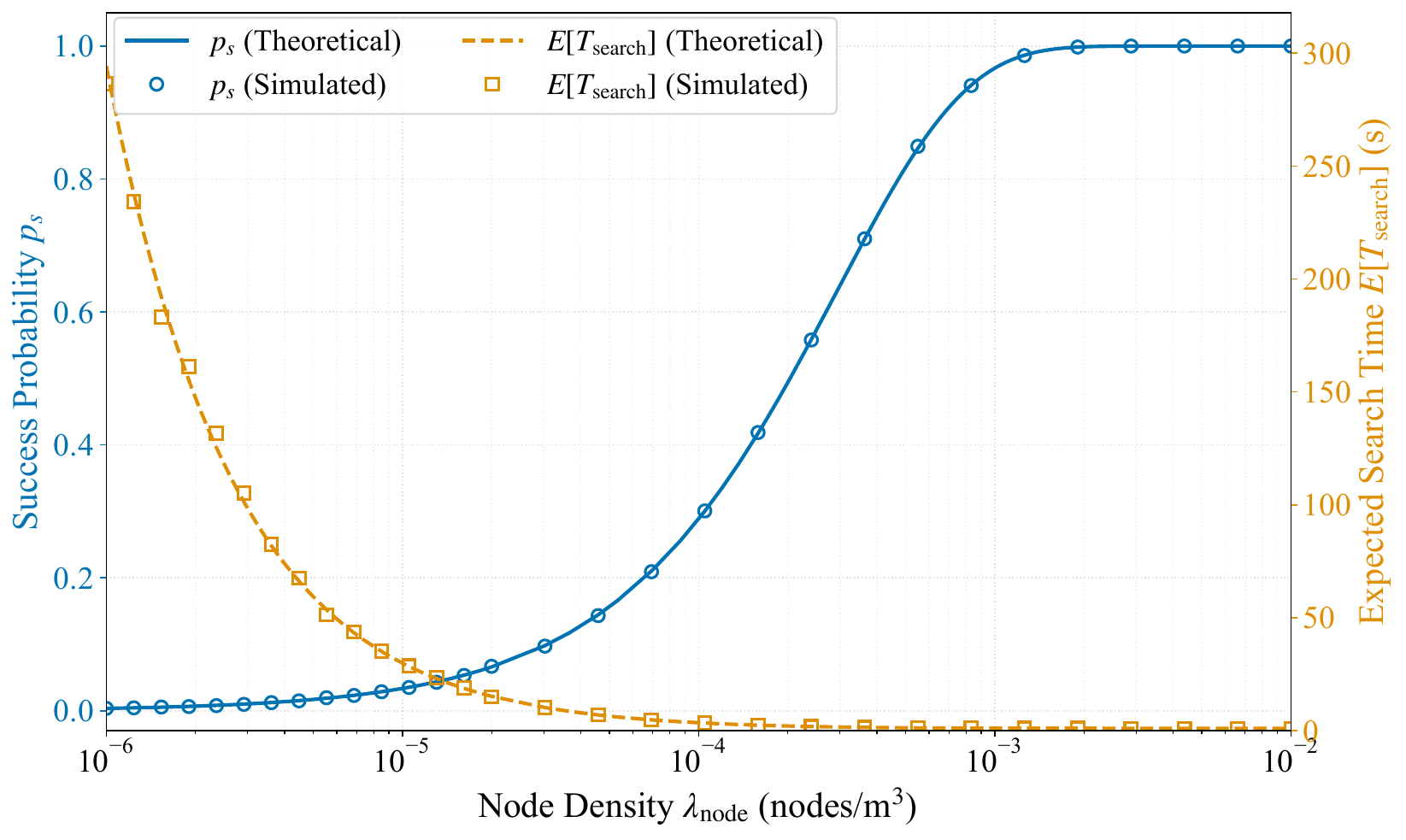}

    \caption{Discovery metrics ($p_s$ and $E[T_{\text{search}}]$) versus node density $\lambda_{\text{node}}$, showing agreement between theory \eqref{eq:ps_snr}, \eqref{eq:tsearch_snr} and simulation.}
    \label{fig:verification_discovery_metrics}
\end{figure}

\textcolor{blue}{Fig.~\ref{fig:verification_discovery_metrics} further validates the PPP-based expressions in \eqref{eq:ps_snr} and \eqref{eq:tsearch_snr}. As the node density increases, the probability of discovering at least one node in a scan increases, while the expected search time decreases. The Monte Carlo results follow the same trend as the analytical expressions, supporting the use of the derived success volume within the stochastic discovery model.}

\subsection{Distance Distribution of the First Discovered Node}
\textcolor{blue}{This subsection refines the discovery analysis by characterizing where the first discovered node is likely to be located. The resulting distance distribution is useful for estimating the subsequent movement and service cost of the AUV.}
Finally, we characterize the random distance $R$ to the first node that is successfully discovered.

\begin{lemma}[Detectable Volume Within Radius $r$]
\label{lemma:Vr_snr}
The volume of the region where a node is both detectable and located within a distance $r$ from the transmitter, denoted $V(r)$, is given by
\begin{align}
V(r)
 = \frac{2\pi}{3}
 \int_{0}^{\phi_{1/2}}
 \min(r, d_{\max}(\theta))^3 
 \sin(\theta)\,\mathrm{d}\theta.
 \label{eq:Vr_snr}
\end{align}
\end{lemma}
\begin{proof}
A node at spherical coordinates $(\rho, \theta, \varphi)$ is within distance $r$ and detectable if its radial coordinate satisfies both $\rho \le r$ and $\rho \le d_{\max}(\theta)$, which is equivalent to $\rho \le \min(r, d_{\max}(\theta))$. Integrating $\rho^2 \sin(\theta)$ over the volume defined by $0 \le \varphi \le 2\pi$, $0 \le \theta \le \phi_{1/2}$, and $0 \le \rho \le \min(r, d_{\max}(\theta))$ yields the result.
\end{proof}

\begin{theorem}[Distribution of the Distance to the Nearest Discoverable Node]
\label{thm:distance_pdf_snr}
\textcolor{blue}{Conditioned on a successful scan, \mbox{$N_{\mathrm{succ}}\geq1$}, let $R$ be the distance to the nearest discoverable node.} Its cumulative distribution function (CDF) and probability density function (PDF) are
\begin{align}
    F_R(r) &= \frac{1 - \exp[-\lambda_{\text{node}}V(r)]}{\textcolor{blue}{p_s}}, \label{eq:cdf_snr}\\
    f_R(r) &= \frac{\lambda_{\text{node}}}{\textcolor{blue}{p_s}}
    \frac{\mathrm{d}V(r)}{\mathrm{d}r}
    \exp[-\lambda_{\text{node}}V(r)], \label{eq:pdf_snr}
\end{align}
\textcolor{blue}{where $p_s=1-\exp(-\lambda_{\text{node}}V_{\text{success}})$ is the probability that the scan contains at least one discoverable node.} The derivative of the volume $V(r)$ is given by
\begin{align}
    \frac{\mathrm{d}V(r)}{\mathrm{d}r}
    = 2\pi r^2
    \left[
    1 - \cos\left(\min(\phi_{1/2},\theta_r)\right)
    \right],
    \label{eq:dvdr_snr}
\end{align}
and $\theta_r$ is the angle implicitly defined by the relation $d_{\max}(\theta_r)=r$. \textcolor{blue}{For $r\geq d_{\max}(0)$, $V(r)=V_{\text{success}}$, so $F_R(r)=1$ and $f_R(r)=0$.}
\end{theorem}
\textcolor{blue}{The conditioning in \eqref{eq:cdf_snr}--\eqref{eq:pdf_snr} makes the distribution consistent with a completed discovery: it describes how far the first discoverable node is likely to be from the AUV, given that the current scan succeeds. This distance directly affects the subsequent travel and servicing cost.}
\begin{proof}
\textcolor{blue}{The event that at least one discoverable node lies within $V(r)$ is a subset of the successful-scan event. The PPP void probability therefore gives}
\begin{align*}
\textcolor{blue}{
P(R\leq r\mid N_{\mathrm{succ}}\geq1)
=\frac{1-e^{-\lambda_{\mathrm{node}}V(r)}}
{1-e^{-\lambda_{\mathrm{node}}V_{\text{success}}}},}
\end{align*}
\textcolor{blue}{which proves \eqref{eq:cdf_snr}; differentiation gives \eqref{eq:pdf_snr}.} To find $\frac{\mathrm{d}V(r)}{\mathrm{d}r}$, we apply the Leibniz integral rule to (\ref{eq:Vr_snr}), $\frac{\mathrm{d}}{\mathrm{d}r} \int_a^b g(x,r) \mathrm{d}x = \int_a^b \frac{\partial g}{\partial r} \mathrm{d}x$. The derivative of $\min(r, d_{\max}(\theta))^3$ with respect to $r$ is $3r^2$ only when $r < d_{\max}(\theta)$, and 0 otherwise. Since $d_{\max}(\theta)$ is a monotonically decreasing function of $\theta$, the condition $r < d_{\max}(\theta)$ is equivalent to $\theta < \theta_r$. The integral for the derivative thus becomes $2\pi r^2 \int_{0}^{\min(\phi_{1/2}, \theta_r)} \sin(\theta) \,\mathrm{d}\theta$, which evaluates to (\ref{eq:dvdr_snr}).
\textcolor{blue}{More explicitly, the radial upper limit in \eqref{eq:Vr_snr} can be differentiated as}
\begin{align}
\textcolor{blue}{
\frac{\partial}{\partial r}\min(r,d_{\max}(\theta))^3
=3r^2\mathbf{1}\{r<d_{\max}(\theta)\}.}
\end{align}
\textcolor{blue}{Because $d_{\max}(\theta)$ decreases with $\theta$ over the scanned beam, this indicator restricts the angular integration to $0\le\theta\le\min(\phi_{1/2},\theta_r)$. Substituting this indicator form into the derivative of $V(r)$ gives}
\begin{align}
\textcolor{blue}{\begin{aligned}
\frac{\mathrm{d}V(r)}{\mathrm{d}r}
&=2\pi r^2\int_0^{\min(\phi_{1/2},\theta_r)}
\sin\theta\,\mathrm{d}\theta \\
&=2\pi r^2\!\left[1-\cos\!\left(\min(\phi_{1/2},\theta_r)\right)\right],
\end{aligned}}
\end{align}
\textcolor{blue}{which proves \eqref{eq:dvdr_snr}.}
\end{proof}

\begin{figure}[!t]
    \centering

    \includegraphics[width=\columnwidth]{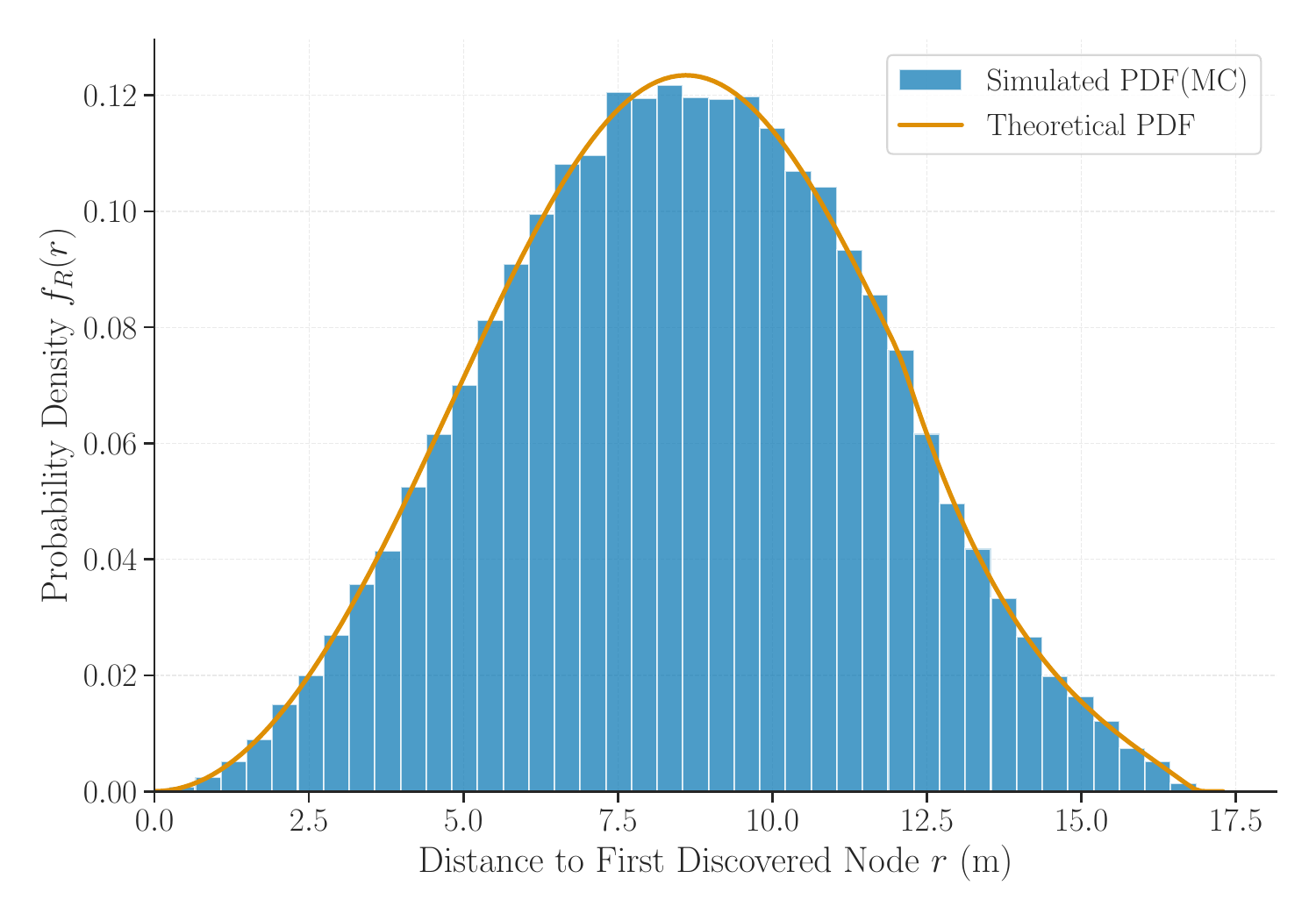}

\caption{\textcolor{blue}{Conditional PDF $f_R(r)$ of the distance to the \roundtwo{nearest discoverable node, given a successful scan ($N_{\mathrm{succ}}\geq 1$)}. The plot compares the analytical PDF \roundtwo{in} \eqref{eq:pdf_snr} with the histogram from \roundtwo{successful Monte Carlo realizations}.}}
    \label{fig:verification_distance_distribution_pdf}
\end{figure}

\textcolor{blue}{Fig.~\ref{fig:verification_distance_distribution_pdf} validates the conditional distance distribution derived in \eqref{eq:pdf_snr}. The analytical PDF is compared with the histogram obtained from successful Monte Carlo realizations of the PPP node field and the SNR-based discovery condition. Their agreement shows that the derived distribution captures the spatial location of the first discovered node once a scan succeeds, which is later used to characterize the service distance and mission-level cost.}

\begin{remark}[Mean Discovery Distance]
The expected distance to the nearest discoverable node, a key parameter for mission planning, \textcolor{blue}{can be evaluated as}
\begin{align}
\textcolor{blue}{E[R]
= \int_{0}^{d_{\max}(0)} r f_R(r) \mathrm{d}r
= \int_{0}^{d_{\max}(0)} (1 - F_R(r)) \mathrm{d}r},
\end{align}
which simplifies to
\begin{align}
\textcolor{blue}{\begin{aligned}
E[R]
&=\frac{1}{p_s}\int_{0}^{d_{\max}(0)}
\bigl\{\exp[-\lambda_{\text{node}}V(r)]\\[-1mm]
&\hspace{25mm}-\exp[-\lambda_{\text{node}}V_{\text{success}}]\bigr\}\,\mathrm{d}r.
\end{aligned}}
\label{eq:mean_distance_snr}
\end{align}
This integral generally lacks a closed form but is easily evaluated numerically. \textcolor{blue}{The subtraction of the failed-scan probability and division by $p_s$ follow from conditioning on $N_{\mathrm{succ}}\geq1$.}
\end{remark}

\section{State-Aware Optimal Point Servicing}
\label{sec:optimal_scheduling}

Upon discovering \textcolor{blue}{an} Rx at an initial distance $R$, the Tx initiates a mobile service phase and moves toward the Rx to deliver wireless information and power. The design of the servicing rule affects both the AUV energy expenditure and the residual energy of the network. A state-agnostic rule, such as always communicating before charging, may be inefficient when the discovered node lacks sufficient energy to support the communication task. \textcolor{blue}{This motivates a state-aware rule that uses the node's reported residual energy to select the subsequent service action.}

In this section, we first compute the transmit power required for a given service task under the physical-layer constraints. Building on this calculation, we introduce the State-Aware Optimal Point Servicing (SA-OPS) policy, \textcolor{blue}{a threshold-based rule} that combines service-location selection with residual-energy feedback.

\subsection{\textcolor{blue}{Physical-Layer Communication Power and Beam Selection}}
\textcolor{blue}{This subsection separates communication-power control from rated-power charging and then selects the beam order used by both tasks. The communication link uses the minimum power required by its received-power threshold, whereas WPT uses the independent rated charging power defined in the service model.}
\label{subsec:phy_opt}

Let $P_{\text{th,com}}$ denote the minimum average received power required for WIT. The relationship between the communication transmit power $P_{\text{Tx,com}}(t)$ and average received power is
\begin{align}
    E[P_{\text{Rx}}(t)] = P_{\text{Tx,com}}(t) \cdot E[H_{\text{move}}(t)], \nonumber
\end{align}
where the expectation is taken over the pointing-error distribution. The minimum communication power satisfying $E[P_{\text{Rx}}(t)]\ge P_{\text{th,com}}$ is
\begin{align}
    P_{\text{Tx,com}}^*(t) = \frac{P_{\text{th,com}}}{E[H_{\text{move}}(t)]}.
    \label{eq:p_tx_star_definition}
\end{align}
From \eqref{eq:p_tx_star_definition}, minimizing communication power is equivalent to maximizing $E[H_{\text{move}}(t)]$. \textcolor{blue}{This is a communication link-budget inversion, not the WPT power rule: WPT uses $P_{\mathrm{Tx,WPT}}$ and benefits from the same gain-maximizing beam order through a higher received charging power.}

{\color{blue}
\begin{proposition}[Lambertian Order under General 2D Pointing Jitter]
\label{prop:optimal_m_2d_full}
Let $\theta_x\sim\mathcal{N}(0,\sigma_x^2)$ and $\theta_y\sim\mathcal{N}(0,\sigma_y^2)$ be independent tangent-plane pointing components, with $\tan\theta=\sqrt{\theta_x^2+\theta_y^2}$. The normalized average gain factor is
\begin{align}
P(m;\sigma_x,\sigma_y)
&=\frac{m+1}{2\pi\sigma_x\sigma_y}
\int_{-\infty}^{\infty}\!\int_{-\infty}^{\infty}\nonumber\\[-1mm]
&\quad\times
\frac{e^{-x^2/(2\sigma_x^2)-y^2/(2\sigma_y^2)}}
{(1+x^2+y^2)^{m/2}}\,\mathrm{d}x\mathrm{d}y.
\label{eq:general_anisotropic_gain}
\end{align}
The implemented order is selected directly from the \roundtwo{adopted configuration set},
\begin{align}
\widehat m^*=\arg\max_{1\le m\le m_{\max}}P(m;\sigma_x,\sigma_y).
\label{eq:implemented_mstar}
\end{align}
\roundtwo{The benchmark uses $m_{\max}=185$ as the upper bound of the adopted configuration range.}\footnote{\roundtwo{Through $m=-\ln2/\ln(\cos\phi_{1/2})$, the upper endpoint corresponds to a half-power semi-angle of approximately $5^\circ$, consistent with the narrow-beam configuration adopted in \cite{7551212}.}}
For the isotropic special case $\sigma_x=\sigma_y=\sigma$, \eqref{eq:general_anisotropic_gain} reduces to
\begin{align}
P(m;\sigma,\sigma)
=\frac{m+1}{\sigma^2}\int_0^{\infty}
u(1+u^2)^{-m/2}
\exp\!\left(-\frac{u^2}{2\sigma^2}\right)\mathrm{d}u.
\label{eq:isotropic_gain_integral}
\end{align}
An interior stationary point, when it exists, satisfies
\begin{equation}
\begin{aligned}
\int_0^{\infty}
&u\left(1-\frac{m+1}{2}\ln(1+u^2)\right)(1+u^2)^{-m/2}\\
&\times\exp\!\left(-\frac{u^2}{2\sigma^2}\right)\mathrm{d}u=0.
\end{aligned}
\label{eq:implicit_mstar_final}
\end{equation}
Under the small-angle approximation, the general anisotropic gain becomes
\begin{align}
P(m;\sigma_x,\sigma_y)
\approx\frac{m+1}{\sqrt{(1+m\sigma_x^2)(1+m\sigma_y^2)}}.
\label{eq:anisotropic_gain_approx}
\end{align}
Over the benchmark \roundtwo{configuration} range, \eqref{eq:anisotropic_gain_approx} increases monotonically with $m$. It therefore provides a compact approximation of the gain trend but no interior optimum in this regime; the implemented order is selected by the constrained exact evaluation in \eqref{eq:implemented_mstar}.

For the benchmark aggregate jitter $\sigma_x^2+\sigma_y^2=0.02$ and ratios $\sigma_y/\sigma_x=1.1$, $1.3$, and $1.5$, numerical Gaussian quadrature gives relative gain changes of $0.19\%$, $1.40\%$, and $3.24\%$, respectively, compared with the isotropic case at $m=185$. In all three cases, \eqref{eq:implemented_mstar} \roundtwo{attains the adopted configuration boundary} at $\widehat m^*=185$.

\begin{proof}
The identity $\cos^m\theta=(1+\theta_x^2+\theta_y^2)^{-m/2}$ and the joint Gaussian density give \eqref{eq:general_anisotropic_gain}. Setting $\sigma_x=\sigma_y=\sigma$ and transforming to polar coordinates gives \eqref{eq:isotropic_gain_integral}; differentiation under the integral sign yields \eqref{eq:implicit_mstar_final}. For small jitter, $\cos^m\theta\approx\exp[-m(\theta_x^2+\theta_y^2)/2]$. Independence then factorizes the expectation, and $\mathbb{E}[e^{-m\theta_j^2/2}]=(1+m\sigma_j^2)^{-1/2}$ yields \eqref{eq:anisotropic_gain_approx}. The constrained order and reported gain changes are obtained by evaluating \eqref{eq:general_anisotropic_gain} over $1\le m\le185$ while holding $\sigma_x^2+\sigma_y^2$ fixed.
\end{proof}
\end{proposition}
}

\subsection{The SA-OPS Policy}
\textcolor{blue}{This subsection combines the physical-layer power calculation with residual-energy feedback to form the SA-OPS decision rule. The policy is organized as a small set of threshold-based cases that determine whether the AUV should charge before communication, communicate and then replenish, or communicate only.}
\label{subsec:saops_policy}

We now introduce the SA-OPS policy. The rule is organized around \textcolor{blue}{three state-dependent threshold choices, following the general resource-management principle of threshold control} \cite{scarf1960optimality}:
\begin{enumerate}
    \item \textcolor{blue}{Service location.} High-power tasks are executed at the point of closest approach ($d_{\text{min}}$), where the channel gain is high.
    \item \textcolor{blue}{Physical-layer beam selection.} The beam configuration (\roundtwo{$\widehat m^*$}) is selected to maximize the expected link gain, thereby reducing the required communication power and increasing the received charging power, as derived in Proposition \ref{prop:optimal_m_2d_full}.
    \item \textcolor{blue}{State-aware logic.} The node's reported residual energy is used to select the service action.
\end{enumerate}

\textcolor{blue}{SA-OPS couples the stochastic discovery analysis with a low-complexity local service rule driven by the residual-energy report of the encountered node. The healthy-energy threshold is calibrated offline from sample-averaged finite-mission KPIs, after which online execution follows a three-branch decision.}
\textcolor{blue}{The derivation of discovery statistics and physical-layer settings is performed offline. Once the task-specific transmit powers and realizable Lambertian order have been obtained, each encounter requires one residual-energy report and at most two scalar threshold comparisons, giving an online complexity of \(\mathcal{O}(1)\) per encountered node. The term optimal point in SA-OPS denotes the closest feasible service point at which the link-budget-compliant \emph{communication} power is minimized within the adopted local geometry.}
SA-OPS divides the service encounter into three sequential phases.

Phase 1 (Guided Approach and Tracking)
Upon initial discovery at distance $R$, the AUV transits towards the Rx. During this approach, it emits periodic, low-power optical "pings". The primary purpose of these pings is not data transfer but active tracking and trajectory refinement. The node's response allows the AUV to continuously verify its heading, ensuring it reaches the intended optimal servicing point at $d_{\text{min}}$.

Phase 2 (State Exchange at Closest Approach)
The AUV continues its guided approach until it reaches the pre-determined point of closest approach, $d_{\text{min}}$. \textcolor{blue}{At this location, where the link-budget-compliant communication power $P_{\text{Tx,com}}^*(d_{\text{min}})$ is minimized within the adopted local geometry, the AUV initiates an acknowledged state-exchange handshake.} It sends a query, and the node responds with its current residual energy, $E_{\text{res}}$.
\textcolor{blue}{The state exchange is modeled as a short acknowledged control exchange, separate from payload WIT. If the encoded query, state report, and acknowledgements contain a total of $L_{\mathrm{hs}}$ bits, their optical airtime is approximately $t_{\mathrm{hs}}=L_{\mathrm{hs}}/R_b$. At the benchmark rate $R_b=1$ Mbps, even $L_{\mathrm{hs}}=10^3$ bits gives $t_{\mathrm{hs}}=1$ ms, four orders of magnitude below the 10-s payload interval and also below the charging intervals, which are typically hundreds of seconds in the simulated service cases. The mission-level benchmark therefore omits handshake airtime, optical-mode switching, and processing energy. For measured hardware values, a policy serving $N_{\mathrm{svc}}$ nodes has the corrections $\Delta T_{\mathrm{ctrl}}=N_{\mathrm{svc}}(t_{\mathrm{sw}}+t_{\mathrm{hs}})$ and $\Delta E_{\mathrm{ctrl}}=N_{\mathrm{svc}}[E_{\mathrm{proc}}+P_{\text{AUV,Platform}}(t_{\mathrm{sw}}+t_{\mathrm{hs}})]$. Here, $t_{\mathrm{sw}}$ is the total optical-mode switching latency per serviced node, and $E_{\mathrm{proc}}$ is the corresponding handshake and control-processing energy. These corrections become relevant when their cumulative contribution is comparable to the payload and charging costs. High-rate short-packet \roundtwo{UWOC} and retransmission mechanisms are reported in \cite{lu201960m,11145354,nguyen2021probing}.}
\textcolor{blue}{The decision rule assumes a successful state exchange after the acknowledgement/retransmission layer. After a bounded number of unacknowledged attempts, a conservative implementation either skips the node or executes the charge-before-communicate branch when the remaining AUV budget permits. A protocol-level extension can incorporate the handshake failure probability into the SA-OPS decision statistics.}

Phase 3 (State-Aware Execution at \texorpdfstring{$d_{\text{min}}$}{Phase 3})
While hovering at or near $d_{\text{min}}$, the AUV executes a decision logic based on the reported $E_{\text{res}}$. The policy, formally defined in Algorithm \ref{alg:saops}, uses two energy thresholds, $E_{\text{comm}}$ (minimum energy to complete the communication task) and $E_{\text{healthy}}$ (target energy for long-term network sustainability).
\textcolor{blue}{The underlying mechanism is a state-dependent allocation of the finite AUV energy budget. The AUV's remaining energy is shared across future encounters: attempting WIT at a critically depleted node may fail, whereas charging an already healthy node provides little marginal mission benefit and reduces the energy available for later service. Accordingly, $E_{\text{comm}}$ acts as a communication-feasibility boundary, while $E_{\text{healthy}}$ limits per-node replenishment once the desired energy condition has been reached. By spending WPT energy only when it restores communication feasibility or materially improves node energy health, SA-OPS reduces avoidable AUV expenditure and reserves energy for subsequent encounters. This converts local residual-energy feedback into a mission-level tradeoff between individual-node recovery and network-wide service coverage, with the objective of slowing network energy deterioration under a fixed AUV budget. Specifically, if $E_{\text{res}}<E_{\text{comm}}$, SA-OPS first charges the node to communication feasibility, performs WIT, and then replenishes it toward $E_{\text{healthy}}$. If $E_{\text{comm}}\le E_{\text{res}}<E_{\text{healthy}}$, the node can already support communication, so SA-OPS completes WIT first and then replenishes it toward the healthy threshold. If $E_{\text{res}}\ge E_{\text{healthy}}$, additional charging is unnecessary for the current mission objective, so SA-OPS communicates only and preserves AUV energy for later encounters.}

{\color{black}
\begin{algorithm}[!ht]
\caption{The SA-OPS Policy}
\label{alg:saops}
\begin{algorithmic}[1]
\State \textbf{Input:} Closest approach distance $d_{\text{min}}$, required communication duration $t_{com}$, energy thresholds $E_{\text{comm}}, E_{\text{healthy}}$.
\Statex
\State \textbf{Phase 1 \& 2: Approach and State Exchange}
\State AUV moves to $d_{\text{min}}$ using tracking pings.
\State At $d_{\text{min}}$, AUV queries node and receives residual energy $E_{\text{res}}$.
\Statex
\State \textbf{Phase 3: State-Aware Execution at $d_{\text{min}}$}
\If{$E_{\text{res}} < E_{\text{comm}}$} \Comment{Case 1: Critically Low Energy}
    \State \textbf{Policy: Charge-then-Communicate}
    \State CHARGE until node energy reaches $E_{\text{comm}}$, using rated $P_{\text{Tx,WPT}}$.
    \State COMMUNICATE for duration $t_{com}$, using $P_{\text{Tx,com}}^*(d_{\text{min}})$.
    \State CHARGE toward $E_{\text{healthy}}$ (or until AUV budget depletion), using rated $P_{\text{Tx,WPT}}$.
\ElsIf{$E_{\text{comm}} \le E_{\text{res}} < E_{\text{healthy}}$} \Comment{Case 2: Sub-Optimal Energy}
    \State \textbf{Policy: Communicate-then-Charge}
    \State COMMUNICATE for duration $t_{com}$, using $P_{\text{Tx,com}}^*(d_{\text{min}})$.
    \State CHARGE until node energy reaches $E_{\text{healthy}}$ (or AUV budget is depleted), using rated $P_{\text{Tx,WPT}}$.
\Else \Comment{Case 3: Healthy Energy ($E_{\text{res}} \ge E_{\text{healthy}}$)}
    \State \textbf{Policy: Communicate-Only}
    \State COMMUNICATE for duration $t_{com}$, using $P_{\text{Tx,com}}^*(d_{\text{min}})$.
\EndIf
\State Mission for this node is complete.
\end{algorithmic}
\end{algorithm}
}

The SA-OPS policy provides three practical effects in the considered local service model:
\begin{enumerate}
    \item \textcolor{blue}{Critical-node support.} It first restores communication feasibility for critically low-energy nodes and, after WIT, continues replenishment toward the healthy threshold (Case 1).
    \item \textcolor{blue}{Network-energy maintenance.} It replenishes intermediate-energy nodes to a target healthy state (Case 2), improving the post-service energy condition of serviced nodes.
    \item \textcolor{blue}{AUV energy conservation.} It avoids unnecessary charging of already-healthy nodes (Case 3), preserving the AUV's limited energy budget for later service encounters.
\end{enumerate}
\textcolor{blue}{Thus, SA-OPS is a local state-aware servicing rule that combines physical-layer power calculation, closest-approach servicing, and residual-energy feedback under the assumptions of the present model.}

\begin{table}[!t]
\centering
\caption{SIMULATION PARAMETERS}
\label{tab:simulation_parameters}
\begin{tabularx}{\linewidth}{X c r}
\toprule
\textbf{Parameter} & \textbf{Symbol} & \textbf{Value} \\
\midrule
\multicolumn{3}{l}{\textit{\textbf{AUV and Mission Parameters}}} \\
AUV velocity & $v$ & \SI{1.5}{m/s} \cite{TeledyneGaviaAUV_misc_2025} \\
Receiver node density & $\lambda_{\text{node}}$ & Varied \\
AUV total energy storage & \textcolor{blue}{$W_{\text{AUV,Storage}}$} & \SI{4.5}{kWh} \cite{TeledyneGaviaAUV_misc_2025} \\
AUV platform power & \textcolor{blue}{$P_{\text{AUV,Platform}}$} & \SI{187.5}{W} \cite{TeledyneGaviaAUV_misc_2025} \\
Dwell time per scan & \textcolor{blue}{$T_{\text{dwell}}$} & \SI{1}{s} \\
Proximity handshake distance & $d_{\text{min}}$ & \SI{1}{m} \\
Lambertian order & $m$ & $1-185$ \cite{7551212} \\
\textcolor{blue}{Service-stage jitter standard deviations} & \textcolor{blue}{$\sigma_x=\sigma_y$} & \textcolor{blue}{$0.1$ rad} \\
Operational depth & \roundtwo{$L_{\text{deep}}$} & \SI{50}{m} \\
\midrule
\multicolumn{3}{l}{\textit{\textbf{Environmental and Noise Parameters \cite{saksvik2025sipm}}}} \\
Total beam attenuation coeff. & $c(\lambda)$ & \SI{0.151}{m^{-1}} \\
Sea-level solar radiation & $E_{\text{sun}}(z=0)$ & \SI{1000}{W/m^2} \\
Solar radiation attenuation coefficient & $\epsilon$ & \SI{0.2}{m^{-1}} \\
Reflectance of solar radiation & $\zeta_r$ & 1.25 \\
\midrule
\multicolumn{3}{l}{\textit{\textbf{Transmitter (Tx) Parameters}}} \\
LED optical power & $P_{Tx}$ & \SI{10}{W} \\
\textcolor{blue}{Rated WPT optical power} & $P_{\mathrm{Tx,WPT}}$ & \textcolor{blue}{\SI{100}{W}} \\
\roundtwo{Discovery} half-power angle & $\phi_{1/2}$ & \SI{60}{\degree} \\
LED wavelength & $\lambda$ & \SI{450}{nm} \\
\midrule
\multicolumn{3}{l}{\textit{\textbf{Receiver (Rx) \cite{saksvik2025sipm}}}} \\
Aperture diameter & $D$ & \SI{30}{cm} \\
Receiver \roundtwo{FoV} (half-angle) & \roundtwo{$\phi_{\text{FoV}}$} & \SI{60}{\degree} \\
Bandpass filter window & $\Delta\lambda$ & \SI{50}{nm} \\
SiPM \roundtwo{photon-detection efficiency} & $\eta$ & 0.31 \\
SiPM gain & $G$ & \textcolor{blue}{$10^6$} \\
SiPM cross-talk & $P_{ct}$ & \SI{8}{\percent} \\
Dark current & $I_d$ & \SI{154}{nA} \\
Excess noise factor & $F$ & \textcolor{blue}{1.2} \\
Load resistance & $R_L$ & \SI{50}{\ohm} \\
Directional dependence factor & $L_f$ & 4 \\
Receiver bandwidth & $B$ & \SI{1}{MHz} \\
Data rate & $R_b$ & \SI{1}{Mbps} \\
\midrule
\multicolumn{3}{l}{\textit{\textbf{Thresholds}}}\\
\textcolor{blue}{SNR-derived discovery threshold} & $P_{\text{th,disc}}$ & \textcolor{blue}{\SI{34.8}{nW}} \\
Communication power threshold & $P_{\text{th,com}}$ & \SI{1.2}{\micro\watt} \\
\textcolor{blue}{Charging-circuit activation power} & $P_{\text{th,charge}}$ & \SI{400}{nW} \\
Energy harvesting efficiency & $\eta_{\text{EH}}$ & \SI{20}{\percent} \cite{rohr2023dive}\\
Node comm. power consumption & $P_{\text{com}}$ & \SI{4}{mW} \cite{11077794} \\
Node sleep power consumption & $P_{\text{sleep}}$ & \SI{80}{\micro W} \cite{11077794} \\
Node energy capacity & $W_{\text{node}}$ & \SI{11286}{J} \cite{rohr2023dive} \\
Min. communication energy  & $E_{\text{comm}}$ & \SI{40}{mJ} \\
\bottomrule
\end{tabularx}
\end{table}


\section{Simulation Results and Analysis}
\label{sec:simulation_results}

In this section, we present simulation results for the proposed SA-OPS policy. We first establish the simulation framework and define the key performance indicators (KPIs). We then examine the sensitivity to node density \textcolor{blue}{within the adopted model}, \textcolor{blue}{compare SA-OPS against four reference policies}, and study the selection of the healthy-energy threshold, $E_{\text{healthy}}$, using multi-criteria decision analysis.

\subsection{Simulation Framework and Performance Metrics}
\textcolor{blue}{This subsection describes the simulation setup used to evaluate the analytical and policy components. We summarize the scenario parameters, baseline strategies, and KPIs before presenting the numerical results.}
\label{sec:setup_and_metrics}

We conduct the performance evaluation through a detailed simulation framework designed to model the \textcolor{blue}{discovery-and-service lifecycle represented in this work}. The simulation parameters, summarized in Table~\ref{tab:simulation_parameters}, are grounded in established models and hardware specifications from prior underwater optical communication and autonomous system research \cite{saksvik2025sipm,rohr2023dive,smart2005underwater,11077794, TeledyneGaviaAUV_misc_2025,7551212}. In this framework, an AUV services \textcolor{blue}{nodes represented by PPP snapshots and treated as quasi-static during each local approach-and-service interval}; initial energy levels follow a truncated normal distribution ($\mu, \sigma$), and the AUV operates under a strict total energy budget. \textcolor{blue}{The platform is assumed to be trimmed near neutral buoyancy, so weight and buoyancy approximately balance and station keeping primarily counters current-induced and other hydrodynamic disturbances \cite{fossen2021handbook,vu2021station}. The aggregate platform-energy cost is accumulated over the search, transit, and service/hover durations, so locomotion expenditure is included at the mission level rather than omitted. \roundtwo{The search-phase energy consumption includes platform operation and discovery transmission:
\begin{equation*}
E_{\mathrm{search}}=
\bigl(P_{\text{AUV,Platform}}+P_{\mathrm{Tx,disc}}\bigr)T_{\mathrm{search}}.
\end{equation*}
Here $P_{\mathrm{Tx,disc}}$ is the discovery transmit power. The platform term includes station keeping and onboard operation during each discovery dwell.} The Gavia platform-power parameter in Table~\ref{tab:simulation_parameters} is used as a manufacturer-based aggregate operating-power term, not as a decomposed model of transit, maneuvering, hovering, and control loads. A phase-resolved power model could change the absolute residual-AUV-energy trajectories and the achievable number of services; current-aware thrust, detailed station-keeping power, and trajectory optimization therefore remain future work.}

We compare SA-OPS with \textcolor{blue}{four reference policies spanning different levels of intervention and state information}. The Communicate-Only strategy performs the required data transfer without energy injection, thereby representing \textcolor{blue}{a low-intervention boundary}. The Always-Charge strategy replenishes every encountered node to $W_{\text{node}}$, representing \textcolor{blue}{a high-intervention boundary}. \textcolor{blue}{Two literature-motivated greedy baselines provide stronger state-aware comparisons. Both poll $K=20$ currently available candidates \roundtwo{in the set $\mathcal{C}_K$} and project each energy to the expected arrival time; hence, they use broader current-state information than SA-OPS, but no future realization. Energy-Deficit Priority (EDP), motivated by energy-urgency scheduling \cite{lei2023urgency}, selects}
{\color{blue}
\begin{align}
    i_{\mathrm{EDP}}^*=\arg\min_{i\in\mathcal{C}_K}E_i^{\mathrm{proj}}.
    \label{eq:edp_rule}
\end{align}
}
\textcolor{blue}{Greedy Utility-per-Joule (UPJ), motivated by charging-utility maximization \cite{ma2018charging}, applies the same state-feasible service sequence as SA-OPS to each candidate and ranks the resulting immediate benefit per AUV service joule:}
{\color{blue}
\begin{equation}
S_i^{\mathrm{UPJ}}=
\frac{r_i+[E_i^{\mathrm{final}}-E_i^{\mathrm{proj}}]^+/W_{\mathrm{node}}+h_i}
{C_i^{\mathrm{svc}}},
\label{eq:upj_score}
\end{equation}
}\textcolor{blue}{where $r_i=1$ when candidate $i$ is initially critical and the service succeeds (and zero otherwise), $h_i=1$ when service moves it from below to at least $E_{\mathrm{healthy}}$ (and zero otherwise), and $C_i^{\mathrm{svc}}$ includes platform and optical-transmission energy during service. The three equal-weight terms reward critical-node rescue, normalized energy delivery, and crossing the healthy threshold. UPJ selects $i_{\mathrm{UPJ}}^*=\arg\max_{i\in\mathcal{C}_K}S_i^{\mathrm{UPJ}}$. EDP emphasizes urgency, whereas UPJ balances urgency and state improvement against service cost. Because the cited methods use different network, mobility, and charging models, EDP and UPJ implement their underlying prioritization principles within the present mission model. Their online ranking cost is $\mathcal{O}(K)$; SA-OPS acts on the single associated node with an $\mathcal{O}(1)$ threshold decision.}
\textcolor{blue}{For the MCDA threshold-selection study, the normalized threshold \(E_{\text{healthy}}/W_{\text{node}}\) is swept from \(0.20\) to \(0.90\) in increments of \(0.01\), giving 71 candidate values for each initial-energy setting. The separate five-policy comparison in Fig.~\ref{fig:six_panel_comparison} uses \(E_{\text{healthy}}=0.40W_{\text{node}}\) as one representative operating point, whereas the MCDA surfaces use the complete 71-point threshold grid.}

Performance is evaluated using a set of KPIs related to sustainability and energy usage. Network resilience is quantified by the network survival time in hours, defined as the post-mission duration until 50\% node failure. \textcolor{blue}{Specifically, \roundtwo{let $N_{\mathrm{node}}$ denote the number of network nodes. If} the final node energies are ordered as \(E_{(1),\mathrm{final}}\leq\cdots\leq E_{(N_{\mathrm{node}}),\mathrm{final}}\), then \(T_{50,\mathrm{post}}=E_{(\lceil N_{\mathrm{node}}/2\rceil),\mathrm{final}}/P_{\mathrm{sleep}}\). This metric starts from the final network state and excludes the elapsed AUV mission time; it measures median post-mission persistence and is used as one input to the subsequent MCDA study.} Operational effectiveness is measured by the critical node rescue efficiency, i.e., the success rate in servicing nodes with critical energy deficits ($E_{\text{res}} < E_{\text{comm}}$). We also report the total energy provided by the AUV and the final variance of network energy to characterize \textcolor{blue}{energy allocation across nodes}.

\begin{figure*}[!t]
    \centering
    \includegraphics[width=\textwidth]{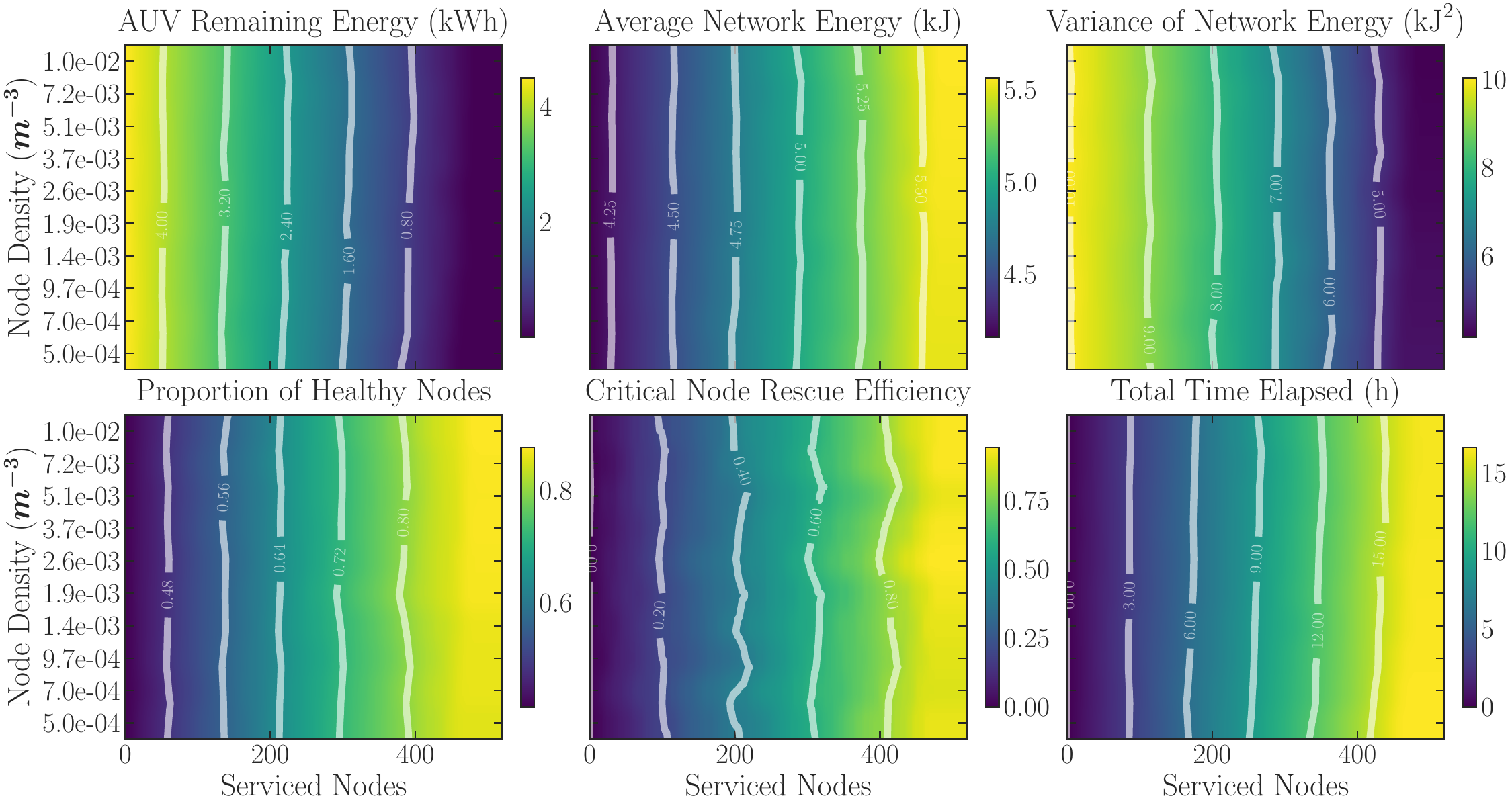}
    \caption{Performance evolution of the SA-OPS policy ($E_{\text{healthy}}=40\% \cdot W_{\text{node}}$) \textcolor{blue}{over node density and completed service events.}}
    \label{fig:density_analysis}
\end{figure*}

\subsection{Results Analysis}
\textcolor{blue}{This subsection analyzes density sensitivity, benchmark comparisons, and MCDA threshold selection by linking each trend to its service mechanism and simulation assumptions.}
\label{sec:density_analysis}

Before conducting a comparative policy evaluation, we first examine the sensitivity of the proposed SA-OPS policy to node density \textcolor{blue}{within the considered single-node servicing model}. Fig.~\ref{fig:density_analysis} presents the evolution of six mission metrics against both node density and the number of completed services. \textcolor{blue}{The dominant progression is along the service-count axis. By approximately 500 completed services, the finite 4.5-kWh AUV budget has been nearly exhausted over about 16.5 h, while the average node energy rises from 4.161 to 5.540 kJ, the energy variance falls from 10.035 to 4.311 kJ$^2$, the healthy-node fraction increases from 41.35\% to 86.90\%, and the critical-node rescue efficiency reaches about 92\%. Thus, the AUV budget is converted into three simultaneous network effects: higher mean stored energy, a more balanced energy distribution, and broader recovery of weak nodes. The nearly monotone AUV-energy and elapsed-time surfaces also provide an engineering budget-to-coverage map: increasing battery capacity, reducing aggregate platform power, or improving WPT efficiency directly increases the number of service events that can be completed before mission termination.}

\textcolor{blue}{Across the tested density range, the values at 500 services change by only about 0.26 h in elapsed time, 0.051 kJ in average energy, 0.281 kJ$^2$ in variance, 2.12 \% points in healthy-node fraction, and 2.88 \% points in rescue efficiency. The mission ledger explains this weak sensitivity. At 500 services, the service dwell accounts for approximately 90.0--96.1\% of the elapsed time and 93.1--97.3\% of the AUV energy expenditure. This stage combines aggregate platform power while the AUV remains near the node with the WIT/WPT transmit cost, whereas density primarily changes the preceding discovery and transit intervals. After association, the same residual-energy rule determines the service action at every density. Consequently, the mission and network-energy surfaces are governed mainly by the accumulated number and type of service events over the evaluated range.}

\textcolor{blue}{Density remains explicit in the end-to-end process through $p_s=1-\exp(-\lambda_{\mathrm{node}}V_{\mathrm{success}})$, so lower density increases the discovery delay before each service event. In the sparse limit, $p_s\approx\lambda_{\mathrm{node}}V_{\mathrm{success}}$ and $E[T_{\mathrm{search}}]\approx T_{\mathrm{dwell}}/(\lambda_{\mathrm{node}}V_{\mathrm{success}})$, allowing search to become the dominant mission cost.}

\textcolor{blue}{At the opposite boundary, the visible-node count \roundtwo{$N_{\mathrm{vis}}$ in a visible volume $V_{\mathrm{vis}}$} is Poisson with mean $\Lambda=\lambda_{\text{node}}V_{\mathrm{vis}}$ and $\Pr\{N_{\mathrm{vis}}\ge2\}=1-e^{-\Lambda}(1+\Lambda)$. When $\Lambda$ is not small, simultaneous replies can create association conflicts or collisions. The present discovery rule selects the strongest received-power candidate and maintains exclusive single-node association throughout service. Multi-node access and scheduling are required for denser regimes in which simultaneous visibility becomes frequent.}

We next compare SA-OPS with the adopted reference policies. Fig.~\ref{fig:six_panel_comparison} tracks the evolution of system performance and illustrates how the \textcolor{blue}{residual-energy-aware service rule} changes the allocation of communication and charging resources. \textcolor{blue}{The six panels should be read jointly because service coverage, mission duration, final energy level, energy balance, healthy-node conversion, and critical-node rescue reward different aspects of a policy.}

\begin{figure*}[!t]
    \centering
    \includegraphics[width=\textwidth]{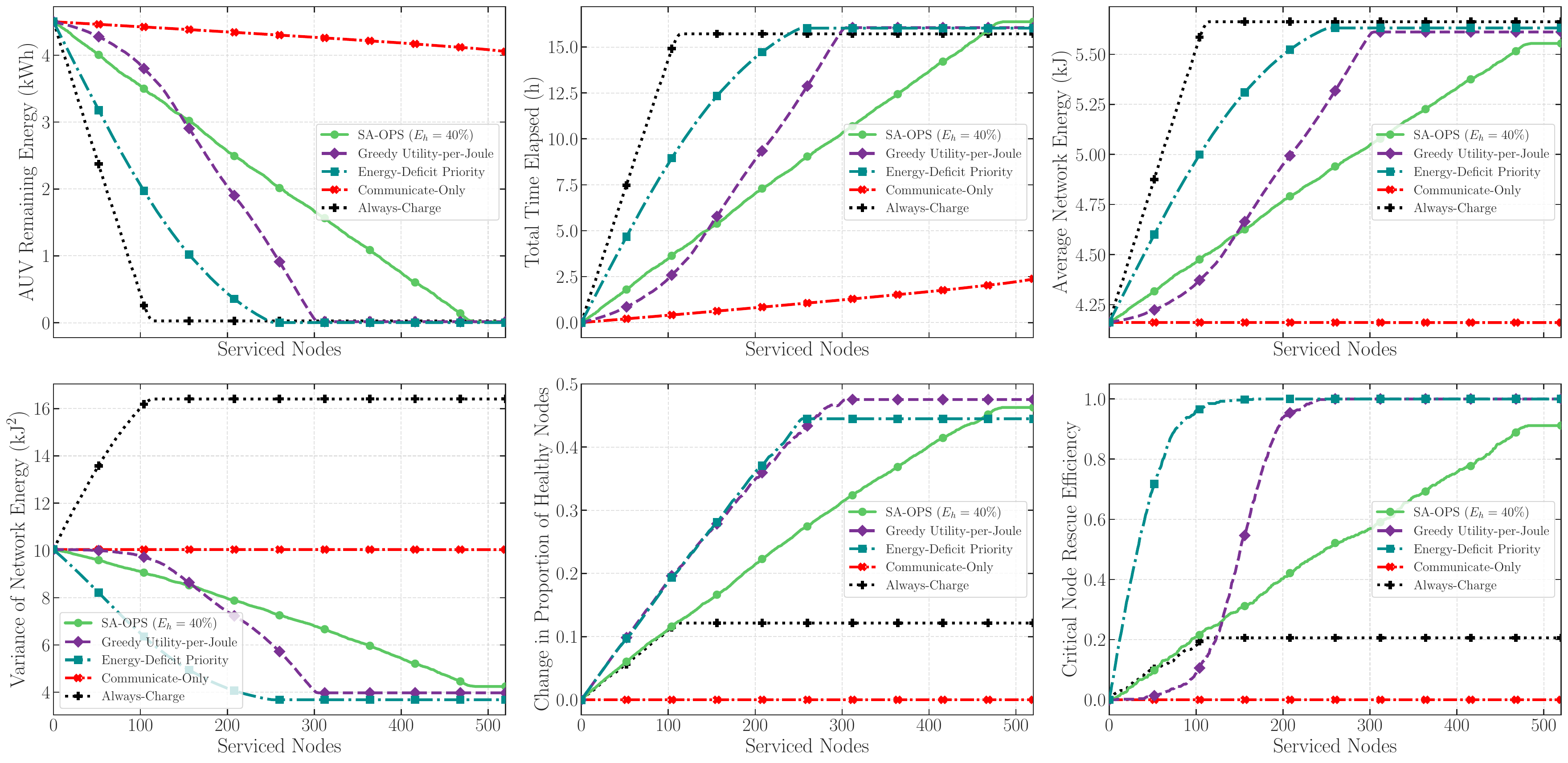}
    \caption{Comparative performance evolution of SA-OPS ($E_{\text{healthy}}=40\% \cdot W_{\text{node}}$) \textcolor{blue}{and four reference policies, averaged over ten Monte Carlo runs, for an initial network state containing 20\% healthy and 10\% critical nodes.}}
    \label{fig:six_panel_comparison}
        \vspace{-3mm}
\end{figure*}

\textcolor{blue}{Fig.~\ref{fig:six_panel_comparison} first exposes the coverage and intervention tradeoff. Communicate-Only reaches all 520 nodes because it spends no energy on WPT, but it neither replenishes energy nor rescues critical nodes. Always-Charge lies at the opposite extreme: full charging consumes the budget after a mean of $110\pm4$ services. EDP directs the AUV to the deepest projected deficits, so its rescue curve rises fastest per completed service and reaches 100\%, but this urgency-first choice terminates after $253\pm2$ services. UPJ also reaches 100\% rescue while extending coverage to $303\pm2$ services because its score discounts benefit by service energy. The service counts are integer-valued in each realization; the reported values are rounded Monte Carlo mean $\pm$ standard deviation.}

\textcolor{blue}{SA-OPS occupies a different information and coverage operating point. It uses only the associated node's energy report, requires neither a network-wide energy map nor $K$-candidate polling, and makes an $\mathcal{O}(1)$ decision. It nevertheless completes $480\pm4$ services (92.3\% coverage) and rescues 91.2\% of initially critical nodes. Because the energy-limited state-aware policies terminate after similar mission durations of approximately 16 h, these counts correspond to overall service throughputs of about 29.3 nodes/h for SA-OPS, 15.8 nodes/h for EDP, and 18.9 nodes/h for UPJ. The lower final rescue fraction of SA-OPS therefore reflects less aggressive prioritization of critical nodes, rather than slower overall service. In return, SA-OPS achieves substantially broader coverage, similar final average energy, lower current-state information requirements, and constant-complexity online decisions.}

\textcolor{blue}{The energy panels further show why no single KPI is sufficient. EDP's concentration on depleted nodes produces the lowest final variance, while UPJ accepts a modest increase in variance to cover roughly 50 additional nodes. Always-Charge yields the largest final average network energy because it performs far fewer repeated search, transit, and communication events and concentrates most of the finite mission budget on charging the comparatively small subset it reaches. The same concentration produces the largest energy variance, only 20.6\% rescue, and the lowest coverage. SA-OPS distributes service over most of the network and therefore provides a more even compromise among coverage, rescue, final average energy, and energy dispersion under limited state information. For panel-wise visualization, a terminated policy's final ordinate is held constant; this graphical continuation does not represent additional services.}

\textcolor{blue}{The threshold $E_{\text{healthy}}$ controls the intervention level. A lower threshold such as 40\% limits charging to more depleted nodes and preserves AUV energy for a larger number of encounters, whereas a higher threshold such as 80\% allocates more energy to each serviced node and can improve local post-service energy reserves at the cost of serving fewer nodes. The appropriate threshold therefore depends on the mission objective; no single value is universally optimal.}

\begin{figure*}[htbp]
    \centering
        \includegraphics[width=0.9\textwidth]{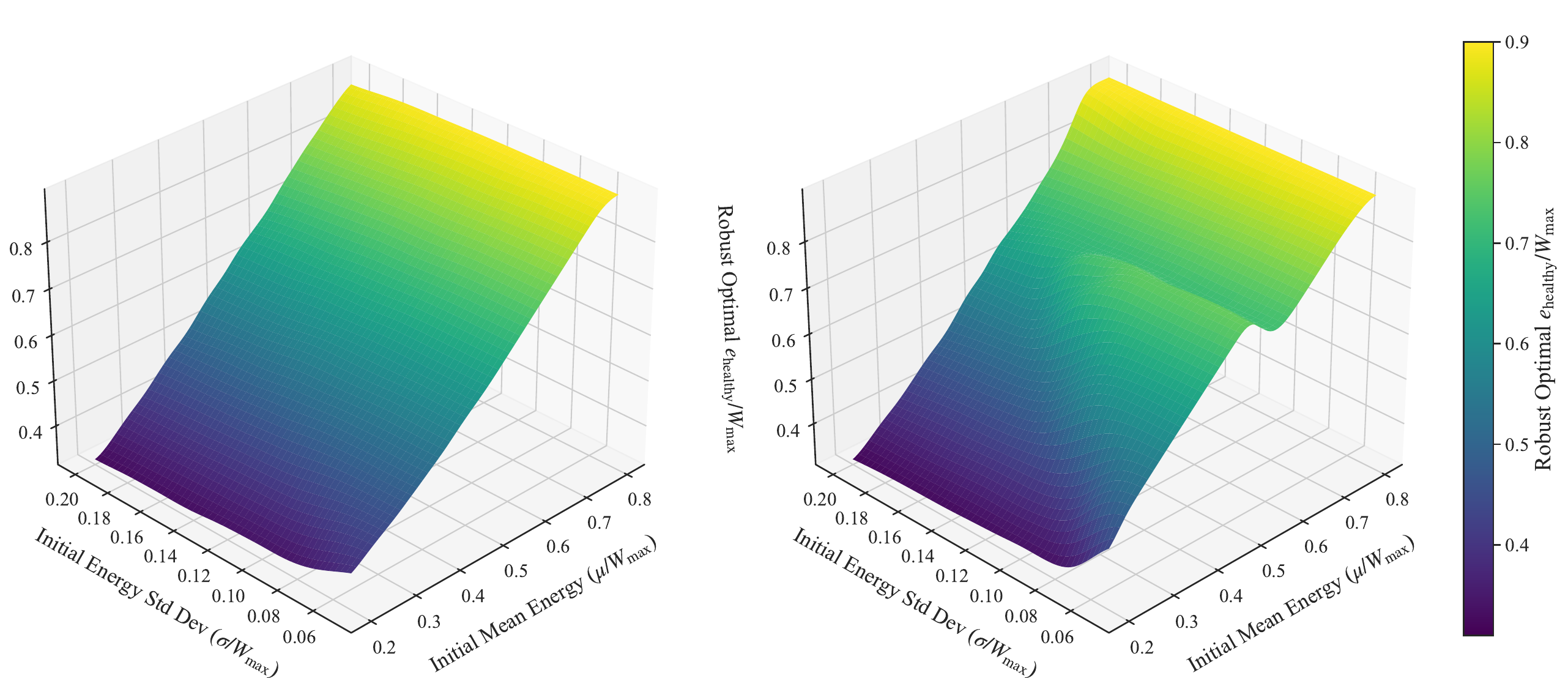}
    \caption{\textcolor{blue}{Robust \(E_{\text{healthy}}\) selected by TOPSIS (left) and LWS (right) for \roundtwo{truncated-normal} initial-energy distributions with varying mean \(\mu\) and standard deviation \(\sigma\).}}
    \label{fig:robust_optimal_surface}
    \vspace{-3mm}
\end{figure*}

\textcolor{blue}{Selecting $E_{\text{healthy}}$ requires a multi-objective tradeoff between network sustainability and AUV energy expenditure. Let $g_e$ denote the SA-OPS rule governed by the normalized threshold $e=E_{\text{healthy}}/W_{\text{node}}$. \roundtwo{We estimate ensemble performance from $N_{\mathrm{MC}}$ finite-mission realizations $\{\boldsymbol{\xi}_n\}_{n=1}^{N_{\mathrm{MC}}}$ using the sample-average KPI vector $\widehat{\mathbf{k}}_{N_{\mathrm{MC}}}(g_e)=N_{\mathrm{MC}}^{-1}\sum_{n=1}^{N_{\mathrm{MC}}}\mathbf{k}(g_e;\boldsymbol{\xi}_n)$. Its components include network survival time, critical-node rescue efficiency, total energy injection, and final energy variance.} The offline threshold-selection objective is}
\begin{align}
    \textcolor{blue}{e^* = \arg\max_{e \in \{0.20,0.21,\ldots,0.90\}} \mathcal{U}\!\left(\widehat{\mathbf{k}}_{N_{\mathrm{MC}}}(g_e)\right)}.
    \label{eq:sample_average_threshold}
\end{align}
\textcolor{blue}{Because $\mathcal{U}$ depends on mission priorities, we use two complementary MCDA methods: linear weighted sum (LWS) \cite{hwang1981methods}, which permits compensation among metrics, and \roundtwo{the technique for order preference by similarity to ideal solution (TOPSIS)} \cite{lai1994topsis}, which ranks alternatives by distance from ideal and anti-ideal points. The candidate grid is \(e\in\{0.20,0.21,\ldots,0.90\}\), giving 71 alternatives for each \((\mu,\sigma)\) setting. Each alternative is evaluated using ten Monte Carlo runs, with the same random initial networks reused across thresholds to reduce comparison noise. For every admissible KPI-weight vector \(\mathbf{w}\), the MCDA method selects one optimal threshold \(e^*_{\mathbf{w}}\); the robust threshold plotted in Fig.~\ref{fig:robust_optimal_surface} is the median of these weight-specific optima.}

The resulting optimal-threshold surfaces, illustrated in Fig.~\ref{fig:robust_optimal_surface}, suggest an asymmetric sensitivity to the mean and standard deviation of the initial energy distribution. \textcolor{blue}{Both methods primarily increase the selected threshold with the initial mean energy \(\mu\), while \(\sigma\) has a secondary effect over most of the tested domain. The TOPSIS surface is smooth: the offset \(e^*-\mu/W_{\text{node}}\) has a mean of 0.122, a median of 0.12, and a 10th--90th percentile interval of 0.10--0.14. LWS has the same median offset of 0.12, but its mean rises to 0.149 and its 90th percentile to 0.26 because of localized high-threshold regions.} This trend is consistent with the state-aware structure of SA-OPS, where decisions are made using individual node feedback. The observation motivates the following affine heuristic:
\begin{align}
    E_{\text{healthy}}^* \approx \mu + \beta \cdot W_{\text{node}}.
    \label{eq:heuristic}
\end{align}
In this formulation, $\beta$ is \textcolor{blue}{a scenario-dependent tuning coefficient} that controls the intervention intensity relative to the node battery capacity. A larger $\beta$ sets a higher target energy threshold, whereas a smaller $\beta$ keeps the threshold closer to the current mean energy and concentrates service effort on lower-energy nodes. \textcolor{blue}{The different surface shapes follow from the aggregation rules. LWS sums independently normalized KPIs and therefore permits a large gain in one metric to compensate directly for deterioration in another. For example, near \((\mu/W_{\text{node}},\sigma/W_{\text{node}})=(0.50,0.10)\), no initially critical nodes are present and rescue efficiency is zero for every threshold; network survival is maximized near \(e=0.76\), delivered energy near \(e=0.89\), and energy variance is minimized near \(e=0.61\). The competing energy and variance objectives cause LWS to switch operating regimes and produce a localized high-threshold region. TOPSIS instead evaluates the joint distance from ideal and anti-ideal KPI vectors and selects a smoother compromise, \(e=0.63\) at that point. Both methods exhibit a dominant dependence on the initial mean energy, while their local threshold selections differ where the KPI optima are widely separated.}

\textcolor{blue}{The affine structure \(E_{\text{healthy}}^*\approx\mu+\beta W_{\text{node}}\), rather than one universal numerical value of \(\beta\), is the main empirical finding. The common median offset suggests \(\beta\approx0.12\) as a first-order calibration for the present benchmark, while the LWS deviations show that this coefficient is not pointwise constant. It summarizes the balance among AUV energy capacity and platform power, optical attenuation and charging efficiency, node energy consumption, traffic demand, the KPI aggregation rule, and the selected mission-level weights; it must therefore be recalibrated when these conditions change. If battery aging is included, the KPI vector can be augmented with a degradation cost such as cumulative depth of discharge or capacity loss \cite{schmalstieg2014aging}, and the resulting coefficient would generally shift because deep-discharge avoidance receives explicit weight.}

\textcolor{blue}{From an implementation perspective, \(\beta\) may be calibrated offline or through short pre-deployment trials, after which the AUV can update \(E_{\text{healthy}}\) using a coarse estimate of the current mean network energy. The observed TOPSIS/LWS difference also provides a practical warning: missions that allow strong compensation among energy-injection and balance objectives may require a lookup table or online MCDA update rather than a single fixed offset. Additional validation is required before transferring the fitted value to substantially different environments, hardware platforms, or mission priorities.}


\section{Conclusion}
This paper presented \textcolor{blue}{an integrated framework} for mobile underwater WIT and WPT centered on the SA-OPS policy. The framework combines an SNR-based analysis of stochastic discovery with \textcolor{blue}{a threshold-based service rule} that uses residual-energy feedback. \roundtwo{Over the tested density range, service dwell dominates mission time and energy, limiting the influence of discovery and transit costs.} The healthy-energy threshold can be approximated by a simple state-dependent heuristic, but its coefficient requires scenario-specific calibration. \roundtwo{In the five-policy experiment, SA-OPS rescues 91.2\% of initially critical nodes and serves about 480 nodes on average. This trades some critical-node rescue for broader coverage than the two information-assisted greedy baselines. SA-OPS requires less current-state information and $\mathcal{O}(1)$ online decisions, while achieving comparable final average network energy.}

\roundtwo{The analysis assumes independent discovery opportunities, aggregate channel and platform-power models, and single-node servicing. Correlated scans require a refined discovery model, while absolute mission estimates depend on platform-specific power consumption. Future work will address these boundaries through correlation-aware discovery, more detailed channel and phase-resolved energy models, and multi-node scheduling.}
\label{sec:conclusion}


\bibliographystyle{IEEEtran}
\bibliography{reference}

\begin{IEEEbiography}
    [{\includegraphics[width=1in,height=1.32in,clip,keepaspectratio]{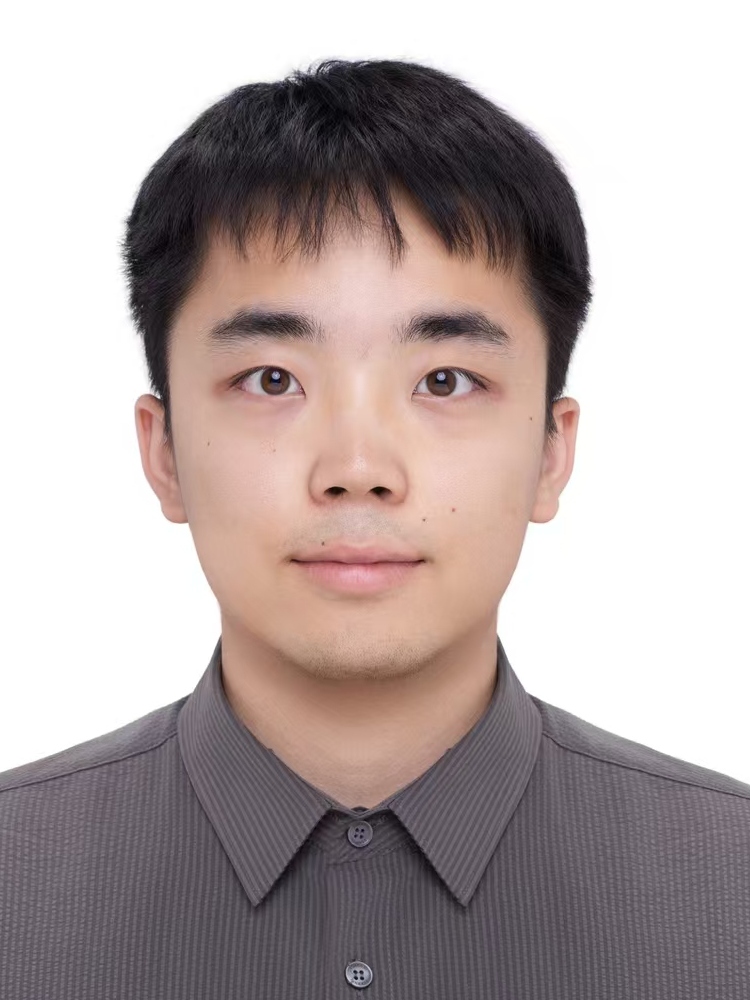}}]{Qiyu Ma}
received the bachelor's degree in electronic engineering from Tsinghua University, Beijing, China, in 2026. He is a visiting student in the Communication Theory Lab at KAUST. His research interests include hybrid underwater acoustic-optical wireless communications, deep learning, and generative artificial intelligence.
\end{IEEEbiography}
\begin{IEEEbiography}
    [{\includegraphics[width=1in,height=1.32in,clip,keepaspectratio]{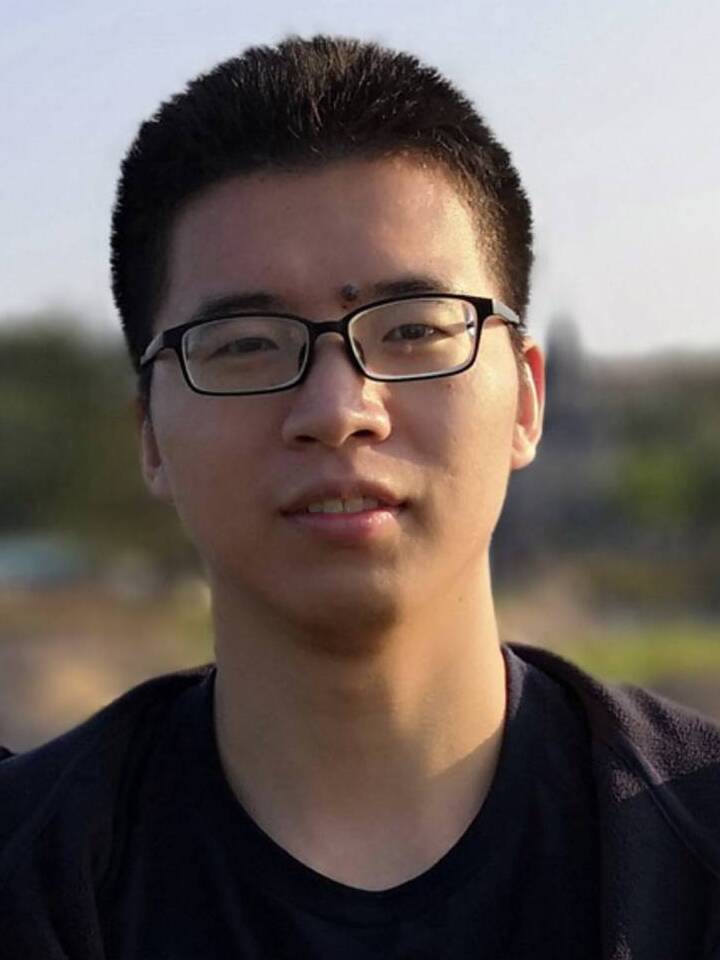}}]{Jiajie Xu}
received the M.Sc. and Ph.D. degrees from Yanshan University and King Abdullah University of Science and Technology (KAUST) in 2019 and 2023, respectively. He is a postdoctoral research fellow in the Communication Theory Lab at KAUST. His research interests include underwater communications and sensing, maritime and space-air-ground-sea integrated networks, stochastic geometry, and energy-harvesting wireless networks.
\end{IEEEbiography}
\begin{IEEEbiography}
    [{\includegraphics[width=1in,height=1.33in,clip,keepaspectratio]{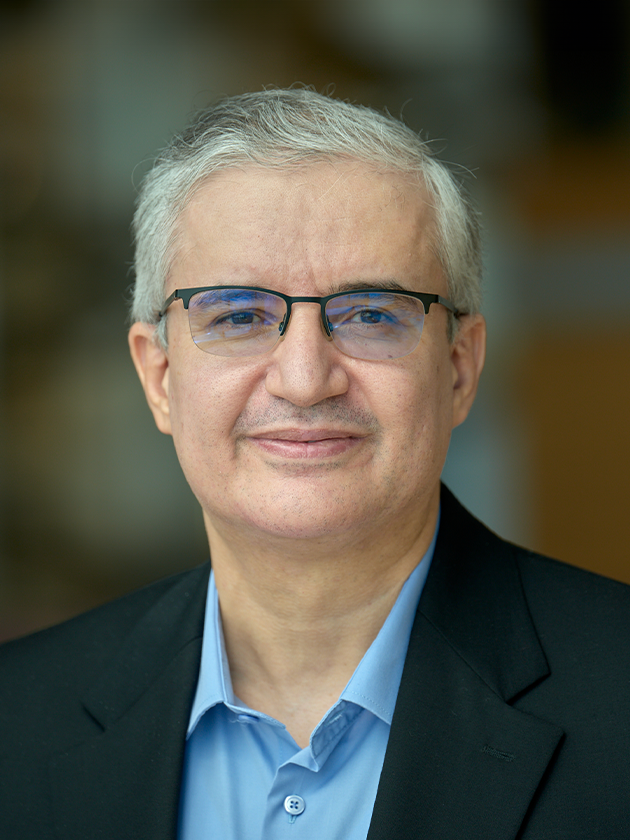}}]{Mohamed-Slim Alouini}
received the Ph.D. degree in electrical engineering from the California Institute of Technology, Pasadena, in 1998. After serving on the faculties of the University of Minnesota and Texas A\&M University at Qatar, he joined KAUST in 2009, where he is a Professor of Electrical Engineering. His research interests include the modeling, design, and performance analysis of wireless communication systems.
\end{IEEEbiography}

\end{document}

%% file: notation_jiajie.tex
\def\mba{{\mathbf{a}}} 
\def\mbb{{\mathbf{b}}}
\def\mbc{{\mathbf{c}}}
\def\mbd{{\mathbf{d}}}
\def\mbe{{\mathbf{e}}}
\def\mbf{{\mathbf{f}}}
\def\mbg{{\mathbf{g}}}
\def\mbh{{\mathbf{h}}}
\def\mbi{{\mathbf{i}}}
\def\mbj{{\mathbf{j}}}
\def\mbk{{\mathbf{k}}}
\def\mbl{{\mathbf{l}}}
\def\mbm{{\mathbf{m}}}
\def\mbn{{\mathbf{n}}}
\def\mbo{{\mathbf{o}}}
\def\mbp{{\mathbf{p}}}
\def\mbq{{\mathbf{q}}}
\def\mbr{{\mathbf{r}}}
\def\mbs{{\mathbf{s}}}
\def\mbt{{\mathbf{t}}}
\def\mbu{{\mathbf{u}}}
\def\mbv{{\mathbf{v}}}
\def\mbw{{\mathbf{w}}}
\def\mbx{{\mathbf{x}}}
\def\mby{{\mathbf{y}}}
\def\mbz{{\mathbf{z}}}
\def\mb0{{\mathbf{0}}}
\def\mb1{{\mathbf{1}}}

\def\mbA{{\mathbf{A}}} 
\def\mbB{{\mathbf{B}}}
\def\mbC{{\mathbf{C}}}
\def\mbD{{\mathbf{D}}}
\def\mbE{{\mathbf{E}}}
\def\mbF{{\mathbf{F}}}
\def\mbG{{\mathbf{G}}}
\def\mbH{{\mathbf{H}}}
\def\mbI{{\mathbf{I}}}
\def\mbJ{{\mathbf{J}}}
\def\mbK{{\mathbf{K}}}
\def\mbL{{\mathbf{L}}}
\def\mbM{{\mathbf{M}}}
\def\mbN{{\mathbf{N}}}
\def\mbO{{\mathbf{O}}}
\def\mbP{{\mathbf{P}}}
\def\mbQ{{\mathbf{Q}}}
\def\mbR{{\mathbf{R}}}
\def\mbS{{\mathbf{S}}}
\def\mbT{{\mathbf{T}}}
\def\mbU{{\mathbf{U}}}
\def\mbV{{\mathbf{V}}}
\def\mbW{{\mathbf{W}}}
\def\mbX{{\mathbf{X}}}
\def\mbY{{\mathbf{Y}}}
\def\mbZ{{\mathbf{Z}}}

\def\mcA{{\mathcal{A}}}
\def\mcB{{\mathcal{B}}}
\def\mcC{{\mathcal{C}}}
\def\mcD{{\mathcal{D}}}
\def\mcE{{\mathcal{E}}}
\def\mcF{{\mathcal{F}}}
\def\mcG{{\mathcal{G}}}
\def\mcH{{\mathcal{H}}}
\def\mcI{{\mathcal{I}}}
\def\mcJ{{\mathcal{J}}}
\def\mcK{{\mathcal{K}}}
\def\mcL{{\mathcal{L}}}
\def\mcM{{\mathcal{M}}}
\def\mcN{{\mathcal{N}}}
\def\mcO{{\mathcal{O}}}
\def\mcP{{\mathcal{P}}}
\def\mcQ{{\mathcal{Q}}}
\def\mcR{{\mathcal{R}}}
\def\mcS{{\mathcal{S}}}
\def\mcT{{\mathcal{T}}}
\def\mcU{{\mathcal{U}}}
\def\mcV{{\mathcal{V}}}
\def\mcW{{\mathcal{W}}}
\def\mcX{{\mathcal{X}}}
\def\mcY{{\mathcal{Y}}}
\def\ncalZ{{\mathcal{Z}}}

\def\mbbA{{\mathbb{A}}}
\def\mbbB{{\mathbb{B}}}
\def\mbbC{{\mathbb{C}}}
\def\mbbD{{\mathbb{D}}}
\def\mbbE{{\mathbb{E}}}
\def\mbbF{{\mathbb{F}}}
\def\mbbG{{\mathbb{G}}}
\def\mbbH{{\mathbb{H}}}
\def\mbbI{{\mathbb{I}}}
\def\mbbJ{{\mathbb{J}}}
\def\mbbK{{\mathbb{K}}}
\def\mbbL{{\mathbb{L}}}
\def\mbbM{{\mathbb{M}}}
\def\mbbN{{\mathbb{N}}}
\def\mbbO{{\mathbb{O}}}
\def\mbbP{{\mathbb{P}}}
\def\mbbQ{{\mathbb{Q}}}
\def\mbbR{{\mathbb{R}}}
\def\mbbS{{\mathbb{S}}}
\def\mbbT{{\mathbb{T}}}
\def\mbbU{{\mathbb{U}}}
\def\mbbV{{\mathbb{V}}}
\def\mbbW{{\mathbb{W}}}
\def\mbbX{{\mathbb{X}}}
\def\mbbY{{\mathbb{Y}}}
\def\nbbZ{{\mathbb{Z}}}

\def\mfrakR{{\mathfrak{R}}}

\def\mrma{{\rm a}}
\def\mrmb{{\rm b}}
\def\mrmc{{\rm c}}
\def\mrmd{{\rm d}}
\def\mrme{{\rm e}}
\def\mrmf{{\rm f}}
\def\mrmg{{\rm g}}
\def\mrmh{{\rm h}}
\def\mrmi{{\rm i}}
\def\mrmj{{\rm j}}
\def\mrmk{{\rm k}}
\def\mrml{{\rm l}}
\def\mrmm{{\rm m}}
\def\mrmn{{\rm n}}
\def\mrmo{{\rm o}}
\def\mrmp{{\rm p}}
\def\mrmq{{\rm q}}
\def\mrmr{{\rm r}}
\def\mrms{{\rm s}}
\def\mrmt{{\rm t}}
\def\mrmu{{\rm u}}
\def\mrmv{{\rm v}}
\def\mrmw{{\rm w}}
\def\mrmx{{\rm x}}
\def\mrmy{{\rm y}}
\def\mrmz{{\rm z}}

\newtheorem{lemma}{Lemma}
\newtheorem{definition}{Definition}
\newtheorem{remark}{Remark}
\newtheorem{theorem}{Theorem}
\newtheorem{proposition}{Proposition}
\newtheorem{corollary}{Corollary}
\newtheorem{example}{Example}
\newtheorem{assumption}{Assumption}

\newcommand{\upcite}[1]{\textsuperscript{\textsuperscript{\cite{#1}}}} 
\newcommand{\romann}[1]{\uppercase\expandafter{\romannumeral #1}} 
\newcommand{\Romann}[1]{\expandafter{\romannumeral #1}} 
\newcommand{\colorb}[1]{{\color{blue} #1}} 
\newcommand{\colorr}[1]{{\color{red} #1}}


\newcommand{\ceil}[1]{\lceil #1\rceil} 

\def\argmin{\operatorname{arg~min}}
\def\argmax{\operatorname{arg~max}}

\def\sinc{{\rm sinc}}
\def\cosc{{\rm cosc}}

\def\larrow{\leftarrow}
\def\rarrow{\rightarrow}
\def\triequ{\triangleq}
\def\simequ{\simeq}

\def\figref#1{Fig.\,\ref{#1}}%
\def\tabref#1{Table\,\ref{#1}}%
\def\equref#1{(\ref{#1})}%
\def\appref#1{Appendix\,\ref{#1}}%
\def\lemref#1{Lemma\,\ref{#1}}%
\def\defref#1{Definition\,\ref{#1}}%
\def\theref#1{Theorem\,\ref{#1}}%
\def\remref#1{Remark\,\ref{#1}}%
\def\secref#1{Sec.\,\ref{#1}}%

\def\assref#1{Assumption\,\ref{#1}}%
\def\ie{{\em i.e.}}
\def\eg{{\em e.g.}}
\def\rme{{\rm e}}
\def\rmd{{\rm d}}
\def\dB{{\rm dB}}
\def\x{\times}
\def\T{\intercal}
\def\H{\dagger}
\def\wbar{\overline}
\def\what{\widehat}
\def\d{{\rm d}}
\def\E{{\mathbb E}}
\def\pd{\partial}
\def\e{{\rm e}}
\def\1{\mathbbmtt{1}}
\def\var{\operatorname{Var}}
\def\cov{\operatorname{Cov}}
\def\mean{\operatorname{mean}}
\def\P{{\mathbb P}}

\def\R{{\mathbb R}}

\def\erfc{\operatorname{erfc}}
\def\erf{\operatorname{erf}}
\def\opt{\mathrm{opt}}

\def\sinr{\mathtt{SINR}}   
\def\snr{\mathtt{SNR}}
\def\sir{\mathtt{SIR}}
\def\scnr{\mathtt{SCNR}}
\def\ase{\mathtt{ASE}}
\def\se{\mathtt{SE}}